\documentclass[journal,twoside,web]{ieeecolor}

\usepackage{generic}
\usepackage{cite}
\usepackage{amsmath,amssymb,amsfonts}
\usepackage{amsmath}
\usepackage{algorithmic}
\usepackage{graphicx}
\usepackage{algorithm,algorithmic}
\usepackage{hyperref}
\hypersetup{hidelinks=true}
\usepackage{textcomp}

\usepackage{rgsMacros}
\usepackage{subcaption}
\let\labelindent\relax
\usepackage{enumitem}
\usepackage{comment}

\usepackage{balance}
\newcommand{\Inc}{\textrm{Inc}}
\newcommand{\AM}{\textcolor{red}}
\allowdisplaybreaks
\newtheorem{exe}{Example}
\newtheorem{corol}{Corollary}
\newtheorem{ass}{Assumption}
\newtheorem{proper}{Property}
\newtheorem{defin}{Definition}
\newtheorem{prob}{Problem}
\newtheorem{cla}{Claim}
\newtheorem{rem}{Remark}
\newtheorem{lem}{Lemma}
\newtheorem{prop}{Proposition}
\newtheorem{thm}{Theorem}
\newtheorem{fct}{Fact}
\newenvironment{lemma}{\begin{lem}}{\hfill $\square$ \end{lem}}
\newenvironment{proposition}{\begin{prop}}{\hfill $\square$ \end{prop}}
\newenvironment{corollary}{\begin{corol}}{\hfill $\square$ \end{corol}}
\newenvironment{example}{\begin{exe}}{\hfill $\square$ \end{exe}}
\newenvironment{remark}{\begin{rem} \rm}{ \end{rem}}
\newenvironment{assumption}{\begin{ass}}{\hfill $\bullet$ \end{ass}}

\newenvironment{theorem}{\begin{thm}}{\hfill $\square$ \end{thm}}
\newenvironment{definition}{\begin{defin}}{ \end{defin}}

\newif\ifitsdraft

\usepackage{pifont}

\begin{document}	

\title{\LARGE \bf Complete Abstractions of Monotone Control Systems: From Model-based to Data-Driven Systems}

\author{Adnane Saoud$^{1}$, Anas Makdesi$^2$, Mohamed Maghenem$^3$,  
Antoine Girard$^4$, and Murat Arcak$^5$
 \thanks{$^1$College of Computing, University Mohammed VI Polytechnic, Benguerir, Morocco (email: adnane.saoud@um6p.ma). $^2$Department of Computer Science, LMU of
Munich, Germany (email: anas.makdesi@lmu.de)
  $^3$Universit\'e Grenoble Alpes, CNRS, Grenoble-INP, GIPSA-lab, F-38000, Grenoble, France (e-mail: mohamed.maghenem@gipsa-lab.fr). $^4$Universit\'e Paris-Saclay, CNRS, CentraleSup\'elec, Laboratoire des signaux et systèmes, 91190, Gif-sur-Yvette, France. (email :antoine.girard@centralesupelec.fr). $^5$University of California, Berkeley, 
 California, USA (e-mail: arcak@berkeley.edu). The work of Murat Arcak was supported in part by the National Science Foundation (NSF) under Grant CNS 2111688.}
}

\maketitle

\begin{abstract}
In this paper, we introduce the \textit{approximate strong upper alternating simulation} 
(ASUAS), a new behavioral relation for transition systems. Building on this relation, we construct \textit{upper-} and \textit{lower-sparse} abstractions for \textit{monotone} systems that together form a \textit{complete} abstraction pair: any controller synthesized for the upper-sparse abstraction can be refined into a controller for the original system, and the absence of a controller for the lower-sparse abstraction implies the absence of a controller for the original system. 
A key feature of our approach is the ability to provably tune the conservativeness gap between the two abstractions by tuning the 
space-discretization parameter. 
We further extend these results, beyond the model-based setting, to \textit{data-driven} systems, where the abstraction is constructed directly from finite sampled data, without requiring an explicit system model. The theoretical results are illustrated through simulations.
\end{abstract}

\section{INTRODUCTION} \label{sec:1}
 The last two decades have shown a growing trend of combining tools from formal methods and control theory. The research at this interface gave birth to a new research area, named  \textit{symbolic control}~\cite{tabuada2009verification,belta2017formal}. The main objective of symbolic control is to make a control system verify logic specifications. To this end, an \textit{abstraction} of the original model (i.e. a dynamical system with a finite number of states and inputs), also named  \textit{symbolic model}, is constructed. Abstractions enable the use of supervisory control techniques \cite{cassandras2009introduction} and (algorithmic) game theory \cite{bloem2012synthesis}.
 
An abstraction is usually constructed to verify a behavioral relationship with respect to the original system, such as the \textit{approximate alternating simulation} (AAS). As a result, under some specifications, the controller designed for the abstraction can be refined into a controller for the original system \cite{tabuada2009verification}. In this case, we say that the abstraction is \textit{sound} under those specifications. For most existing behavioral relationships, the obtained abstractions suffer from \textit{conservativeness} issues \cite{tabuada2009verification,belta2017formal}, since some trajectories of the obtained abstractions are not trajectories of the original system. Hence, if we cannot find a controller for the abstraction, we fail to decide whether or not a controller exists for the original system. This is in contrast to \textit{complete} abstractions, for which the existence of a controller is equivalent to the existence of a controller for the original system.

Various approaches are proposed in the literature to construct abstractions for different classes of systems. Based on establishing one-sided AAS relations, results apply to stabilisable systems \cite{tabuada2008approximate}, incrementally forward-complete systems \cite{zamani2012symbolic}, monotone systems \cite{meyer2015adhs}, 
mixed-monotone systems \cite{coogan2015efficient}, differentially-flat systems \cite{liu2012reactive}, impulsive systems \cite{swikir2020symbolic}, hybrid systems \cite{da2021symbolic}, and networked control systems \cite{zamani2017symbolic}. Based on establishing approximate bisimulation-type relations, which provide two-sided behavioral equivalence guarantees, results consider general nonlinear systems \cite{pola2008approximately}, switched systems  \cite{girard2010approximately}, time-delay systems \cite{pola2010symbolic}, and  stochastic systems \cite{zamani2014symbolic}. This two-sided guarantee, however, comes at a price: all the aforementioned bisimulation-based approaches require the system to satisfy some form of incremental stability property~\cite{angeli2002lyapunov}.

While these model-based techniques are effective in settings where accurate models are available, they become challenging when such models are unavailable or impractical. This motivates the shift towards data-driven approaches, which  build abstractions directly from data. Several approaches have been proposed to construct data-driven abstractions, differing both in the class of systems they target and in the type of guarantee they provide. For nonlinear systems, \cite{hashimoto2022learning} learns the unmodeled dynamics via a Gaussian process and provides deterministic guarantees, whereas \cite{devonport2021symbolic} follows a Probably Approximately Correct (PAC) learning approach. For stochastic systems, \cite{lavaei2020formal} deals with continuous-space Markov decision processes through model-free reinforcement learning with probabilistic guarantees, whereas \cite{badings2023robust} adopts a scenario approach to construct interval Markov decision process abstractions with PAC guarantees. Finally, for monotone systems, \cite{makdesi2023data} provides deterministic guarantees.

The key motivation behind this work stems from a fundamental limitation of classical behavioral relations, such as the AAS relation. These relations, rooted in the formal methods literature~\cite{baier2008principles}, are universal, i.e., they do not target systems with specific structural properties. They can, thus, be used to construct sound abstractions for broad classes of systems and specifications~\cite{tabuada2009verification}.
However, this universality comes at the cost of generating conservative abstractions~\cite{girard2010approximately}.
Motivated by the latter drawback, and inspired by classical AAS relations, we introduce a new behavioral relation, named \textit{approximate strong upper alternating simulation} (ASUAS). This relation is suitable for monotone systems, as we show that any controller synthesized for the abstraction can, under certain specifications, be refined into a controller for the original system. Building on this, we construct \textit{upper-} and \textit{lower-sparse abstractions} that together form a \emph{complete abstraction pair}. Completeness here is understood in the following precise sense: the existence of a controller for the upper-sparse abstraction guarantees the existence of a refined controller for the original system, and the absence of controller for the lower-sparse abstraction certifies the absence of controller for the original system. 
We further provide a constructive approach to tune the conservativeness gap between the two abstractions: for any prescribed error bound, we derive a space-discretization parameter that provably guarantees it. These model-based results are then extended to a data-driven setting, where abstractions are constructed directly from a finite set of data points sampled from the system's trajectories. We introduce data-driven methods to build both upper- and lower-sparse abstractions for monotone systems. We further show how to control the conservativeness between the model-based and data-driven abstractions, providing both probabilistic and deterministic guarantees. 

\textbf{Related work:} The construction of complete symbolic abstractions is, in general, intractable. To the best of our knowledge, the only class of systems for which complete abstractions can be constructed is that of incrementally stable systems~\cite{pola2008approximately,girard2010approximately}, which forms a strict subclass of stable systems. A different approach was recently proposed in~\cite{liu2017robust}, where a complete abstraction is computed for a perturbed version of the original system. However, the absence of a controller for the perturbed system does not imply the same for the original system, leaving the question of completeness with respect to the original system unanswered. Regarding monotone systems specifically, the existing literature offers only sound abstractions, both in the model-based setting~\cite{coogan2015efficient,meyer2015adhs,kim2017symbolic} and in the data-driven setting~\cite{makdesi2023data}. The present work goes strictly beyond the aforementioned results by providing, for the first time, the concept of \emph{complete abstraction pairs} for this class of systems.

 The remainder of the paper is organized as follows. Section~\ref{sec:2} recalls the necessary preliminaries. Section~\ref{sec:3} introduces the ASUAS relation and develops the associated controller refinement procedure. Section~\ref{sec:4} presents the construction of complete abstraction pairs for monotone systems in the model-based setting. Section~\ref{sec:dd} extends these results to the data-driven setting. Finally, Section~\ref{sec:6} illustrates the proposed framework through numerical examples.

\textbf{Notation.} Given a set $Y \subseteq \mathbb{R}^n$, $Y^*$ denotes the set of finite sequences of elements in $Y$, and $Y^w$ denotes the set of infinite sequences of elements in $Y$. For $x \in \mathbb{R}^n$, $|x|$ denotes its infinity norm, and, for $\varepsilon \in \mathbb{R}^n_{+}$, $\Omega_{\varepsilon}(x) := \{z \in \mathbb{R}^n : |z_i-x_i| \leq \varepsilon_i, ~ \forall i \in \{1,2,...,n\}\}$. Similarly, for a set $X \subseteq \mathbb{R}^n$, $\Omega_{\varepsilon}(X) := \bigcup\limits_{x \in X} \Omega_{\varepsilon}(x)$. For an infinite sequence 
$\sigma_y := y_0,y_1,\ldots \in Y^w$ and for $\varepsilon \in \mathbb{R}^n_{+}$, $\Omega_{\varepsilon}(\sigma_y) := \{\sigma_z=z_0,z_1,\ldots \in Y^w : z_i \in \Omega_{\varepsilon}(y_i), ~ \forall i \in \mathbb{N} \}$. For a set $K \subseteq Y^w$, $\Omega_{\varepsilon}(K) := \bigcup\limits_{\sigma_y \in K}\Omega_{\varepsilon}(\sigma_y)$. The natural projection $\pi_A: A \times B \rightarrow A$ is defined by $\pi_A(a,b) = a$. We use the classical component-wise partial order on $\mathbb{R}^n$, defined for $x,y \in \mathbb{R}^n$ by $x \leq y$ if and only if $x_i \leq y_i$ for all $i \in \{1,2,\ldots,n\}$. If neither $x \leq y$ nor $y \leq x$ holds, we say that $x$ and $y$ are \textit{incomparable}. The set of all incomparable couples in $\mathbb{R}^n$ is denoted by $\Inc_{\mathbb{R}^n}$. The sets of minimal elements of $A \subseteq \mathbb{R}^n$ is defined as 
$$\min(A) := 
\{x \in A : \forall z \in A, ~ x \leq z~\text{or}~ (x,z)\in \Inc_{\mathbb{R}^n}\}.$$ For $a\in \mathbb{R}^n$, we let $\downarrow a := \{x \in \mathbb{R}^n : x \leq a\}$ and $\uparrow a := \{x \in \mathbb{R}^n : x \geq a\}$. The lower and upper closures of $A \subseteq \mathbb{R}^n$ are $\downarrow A := \bigcup\limits_{a \in A} \downarrow a$ and $\uparrow A := \bigcup\limits_{a \in A} \uparrow a$, respectively. A subset $A \subseteq \mathbb{R}^n$ is said to be lower closed if $\downarrow A = A$, and upper closed if $\uparrow A = A$. A partial ordering $x \leq y$ between a pair of infinite sequences $x \in A^w$ and $y \in A^w$ holds if and only if $x_i \leq y_i$ for all $i \in \mathbb{N}$. An illustration of the concepts of lower-closed sets and the lower closure of a point is provided in Figure~\ref{fig:lower_closed}.
\begin{figure}[!t]
\begin{center}
\includegraphics[scale=0.5]{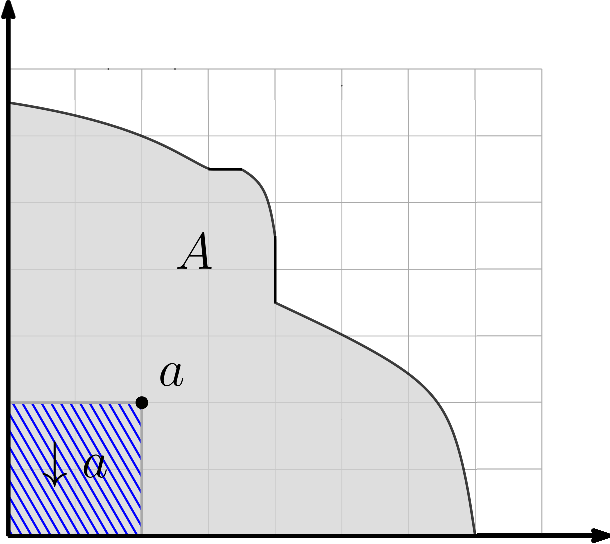}
\end{center}
\caption{A lower-closed set $A \subseteq \mathbb{R}_{+}^2$, with the standard ordering, in gray. The lower closure of a point $a \in A$ is presented in dashed blue.}
\label{fig:lower_closed}	
 \end{figure}

\section{Preliminaries} \label{sec:2}

\subsection{Transition systems}

We start reviewing the notion of \textit{transition system}~\cite{tabuada2009verification}.  
\begin{definition}
[Transition system]
\label{definition6}
	A transition system is a tuple $S :=(X,X^o,U,\Delta,Y,H)$, where $ X $ is the set of states, $X^o \subseteq X$ is the set of initial states, $U$ is the set of inputs, 
 $ \Delta \subseteq X\times U\times X $ is the transition relation, $Y$ is the set of outputs and $H$ is the output map. When $ (x,u,x')\in\Delta$, we use the alternative representation  $ x'\in \Delta(x,u)$,  where the state $x'$ is called a successor of $x$ under the input $u$. 
\end{definition}
The transition system is \textit{finite} (or \textit{symbolic}), if the sets $X$ and $U$ are finite. It is \textit{deterministic}, if there exists at most one successor for any $x \in X$ and $ u \in U $. The set of enabled (admissible) inputs for $x \in X$ is denoted by $U^a(x)$ and defined as $U^a(x) := \{u\in U : \Delta(x,u)\neq \emptyset\}$. 

For the transition system $S$, we assume that for all $x \in X$ and for all $u \in U$, $\Delta(x, u) \neq \emptyset$. This means that for any state all the inputs are enabled. A trajectory of the transition system $S$ is a sequence $\sigma:=(x_0,u_0),(x_1,u_1),...$, where $x_0 \in X^0$, $x_{i+1} \in \Delta(x_i,u_i)$ for all $i \in \mathbb{N}$. The output behavior associated with the trajectory $\sigma$ is the sequence $\sigma_y: = y_0,y_1,...$, where $y_i := H(x_i)$ for all $i \in \mathbb{N}$. We use $\mathcal{B}(S) \subseteq Y^w$ to denote the set of all possible output behaviors of the system $S$.

\subsection{Lower-Closed Specifications}

For the transition system $S :=(X,X^o,U,\Delta,Y,H)$, an output specification $\phi \subseteq Y^{w}$ encodes a set of desirable output behaviors.  The system $S$ is said to satisfy the specification $\phi$ if $\mathcal{B}(S) \subseteq \phi$. A lower-closed specification $\phi$ corresponds to a lower-closed  subset of $Y^w$. Lower-closed specifications can describe complex specifications ranging from safety to reachability, stability, and liveness for lower-closed sets. A fragment of linear temporal logic specifications that are lower closed is  identified in~\cite{kim2016directed}. In the rest of the paper, we focus on lower-closed specifications; analogous results can be formulated for upper-closed specifications.

\subsection{Monotone control systems}

Consider the discrete-time control system
\begin{equation}
\label{dis_sys}
\Sigma : x(k+1) = f(x(k),u(k),d(k)),
\end{equation}
where $x \in X \subseteq \mathbb{R}^n$ is the state, $u \in U\subseteq\mathbb{R}^m$ is the input, and $d \in D \subseteq \mathbb{R}^p$ is a disturbance. 

\begin{definition}[Monotone systems \cite{smith2008monotone}]
The control system $\Sigma$ is monotone if, for each $x_1,x_2 \in X$, for each $u_1,u_2 \in U$, and for each $d_1,d_2\in D$,  $x_1 \leq x_2$, $u_1 \leq u_2$, and $d_1\leq d_2$ then  $f(x_1,u_1,d_1) \leq f(x_2,u_2,d_2)$.
\end{definition}

We require a specific geometry for the sets $U$ and $D$.

\begin{assumption} \label{assum:contr_dist}
The set $U$ is an interval of the form $U :=[\underline{U},\overline{U}]$ with $\underline{U}$, $\overline{U} \in \mathbb{R}^m$, and the set $D$ is a finite union of intervals of the form  $D := \bigcup_{m=1}^M[d_1^m,d_2^m]$ with $d_{1}^1, d_{2}^1, ..., d_1^M, d_{2}^M  \in \mathbb{R}^p$ for some $M \in \mathbb{N}$ \footnote{Without loss of generality, we also assume that for all $i,j \in \{1,2,\ldots,M\}$, with $i \neq j$, we have $[d_1^i,d_2^i] \not\subseteq [d_1^j,d_2^j]$}.
\end{assumption}

The transition-system representation of $\Sigma$ is given, under Assumption \ref{assum:contr_dist}, by\footnote{In the rest of the paper, the discrete-time system $\Sigma$ and the corresponding transition system $S_\Sigma$ can be used interchangeably.} 
\begin{align} \label{eqSsigma}
S_\Sigma :=(X,X^o,U,\Delta,Y,H),
\end{align}
where $X^o \subseteq X$ is the set of initial states, $Y = X$ is the output set, and $H: X \rightarrow Y$ is the identity map. Finally, the transition relation $\Delta$ is defined for $x \in X$ and $u \in U$ as $x' \in \Delta(x,u)$ if and only if there exists $m \in \{1,2,\ldots,M\}$ and $d \in [d_1^m,d_2^m]$ such that $x'=f(x,u,d)$.

When $f$ is continuously differentiable, Kamke-Muller's sufficient condition for monotonicity  is given by
$$ \frac{\partial f_i}{\partial x_j} \geq 0, ~~ \frac{\partial f_i}{\partial u_k} \geq 0, ~~ \frac{\partial f_i}{\partial d_l} \geq 0$$ 
for all $i,j \in \{1,2,\ldots,n\}$, for all $k\in \{1,2,\ldots,m\}$, and for all $l \in \{1,2,\ldots,p\}$; see \cite{smith2008monotone}.

\begin{remark}
Monotone systems constitute a broad and practically relevant class of dynamical systems, with applications in biology~\cite{angeli2003monotone}, traffic flow~\cite{coogan2017formal}, microgrids~\cite{zonetti2019decentralized}, and autonomous vehicles~\cite{smith2008monotone}, among others. At first glance, monotonicity may appear restrictive, as it requires the off-diagonal entries of the Jacobian $\partial f/\partial x$ to be nonnegative, i.e., the system to be cooperative with respect to the positive orthant. This requirement is, however, less stringent than it seems. First, the positive orthant is only one of many admissible order cones: a system is monotone with respect to some orthant whenever its off-diagonal partial derivatives are sign-definite and the associated signed interaction graph contains no negative undirected cycle, a condition that can be verified graphically~\cite{smith2008monotone,angeli2003monotone}; more general changes of coordinates further enlarge the class of systems admitting a monotone representation. Second, physical structure can expose monotonicity that is not apparent in the original coordinates: restricting the dynamics to the invariant subspaces induced by conservation laws, or re-expressing it in suitable reaction coordinates, reveals monotone structure in several biomolecular signaling cascades~\cite{angeli2010graph}. Systems that are not monotone in any of these senses can often be decomposed into an interconnection of monotone components~\cite{dasgupta2007algorithmic}.

\end{remark}

\section{Behavioral relationships} \label{sec:3}

In this section, we introduce the proposed behavioral relationships together with its associated control refinement procedure.

\subsection{The ASUAS relationship}
Consider the two transition systems
\begin{equation} 
\label{eqS1S2}
\begin{aligned}
S_1 & :=(X_1,X_1^o,U_1,\Delta_1,Y_1,H_1) 
\\
S_2 & :=(X_2,X_2^o, U_2,\Delta_2,Y_2,H_2)
\end{aligned}
\end{equation}
such that the following assumption holds. 
 \begin{assumption} \label{assPO}
 The sets $U_1$ and $U_2$ are compact, and subsets of the same set $U \subseteq \mathbb{R}^p$, and $Y_1$ and $Y_2$ are subsets of the same set $Y \subseteq \mathbb{R}^m$.
\end{assumption}
\begin{definition} \label{Def:altsimu_up}
 Consider the transition systems $S_i:=(X_i,X_i^o,U_i,\Delta_i,Y_i,H_i)$, $i=1,2$. Under Assumption \ref{assPO}, and for some $\varepsilon \in \mathbb{R}^m_{+} $, a relation $\mathcal R  \subseteq X_1 \times X_2$ is said to be an $\varepsilon$-approximate strong upper alternating simulation relation, or $\mathcal{R}$ is an $\varepsilon$-ASUAS, from $S_2$ to $S_1$, if:
	\begin{itemize}
		\item[(i)] $\forall x_2^o \in X_2^o$, $\exists x_1^o \in X_1^o$ such that $(x_1^o, x_2^o) \in \mathcal{R}$;
		\item[(ii)]  $\forall (x_1,x_2) \in \mathcal{R}$, $H_1(x_1)\leq H_2(x_2)+\varepsilon$;
        \item[(iii)] $U_2 \subseteq \{\uparrow U_1\}$,  and  $\forall (x_1,x_2) \in \mathcal{R}$, $\forall u_2 \in U^a_2(x_2)$, we have $\{\downarrow u_2\} \cap U_1 \subseteq U^a_1(x_1)$;
		\item[(iv)] $\forall (x_1,x_2) \in \mathcal{R}$, $\forall u_2 \in U^a_2(x_2)$,
		$\forall u_1 \in U^a_1(x_1) \cap \{\downarrow u_2\}$, we have 
 \begin{align} \label{eqiv}
      \hspace{-0.5cm}  \forall x_1'\in  \Delta_1(x_1,u_1), ~ \exists x_2'\in  \Delta_2(x_2,u_2) : (x_1',x_2') \in \mathcal{R}.
        \end{align}
	\end{itemize}
	  In this case, we say that $S_2$ $\varepsilon$-approximately strongly upper alternatingly simulates $S_1$, or $S_2$ $\varepsilon$-ASUAS $S_1$, denoted $S_2 \preccurlyeq^{\varepsilon}_u S_1$. 
\end{definition}




\begin{remark}
\label{rk:existence}
Condition (iii) in Definition \ref{Def:altsimu_up} is mainly used to ensure that, for all $(x_1,x_2) \in \mathcal{R}$ and for all $u_2 \in U^a_2(x_2)$, the set $U^a_1(x_1) \cap 
\{\downarrow u_2\}$ used in condition (iv) is nonempty.
\end{remark}

\begin{remark}
 The difference between the proposed relation and traditional ones in symbolic control~\cite{tabuada2009verification, belta2017formal,baier2008principles} resides in (ii). Indeed, while traditional relations impose \emph{output closeness}, requiring the outputs of related states to remain within $\varepsilon$ distance to each other, the proposed relation imposes only output ordering. Moreover, the proposed relation differs from the upper alternating simulation relation in~\cite{kim2017symbolic} at two levels:
 \begin{itemize}
    \item While the relation in~\cite{kim2017symbolic} requires the output of one system to be upper bounded by the output of the other, we relax this upperbound to a precision $\varepsilon \in \mathbb{R}^m_{+}$.
    
    \item While the relation in~\cite{kim2017symbolic} requires, for each input $u_2 \in U_2^a(x_2)$, the existence $u_1 \in U_1^a(x_1)$ such that \eqref{eqiv} holds,  we require that \eqref{eqiv} holds for each $u_2 \in U_2^a(x_2)$ and for each $u_1 \leq u_2$.
    
    The reason behind this choice is to allow the controller of $S_1$, refined from (an abstraction) $S_2$, to be more permissive, i.e., to admit a larger set of valid control inputs at every state. 
\end{itemize} 
\end{remark}

The following result establishes the transitivity of ASUAS.

\begin{proposition}
\label{prop:transitivity}
 Let $S_i=(X_i,X_i^o,U_i,\Delta_i,Y_i,H_i)$, $i=\{1,2,3\}$ be a collection of transition systems such that $U_1$, $U_2$ and $U_3$ are compact subsets of the same partially ordered set $U \subseteq \mathbb{R}^p$, and $Y_1$, $Y_2$ and $Y_3$ are subsets of the same partially ordered set $Y \subseteq \mathbb{R}^m$. Assume further that $U_1 \subseteq U_2$ or $U_3 \subseteq U_2$, and consider $\varepsilon_1,\varepsilon_2 \in \mathbb{R}^m_{\geq 0}$. Then, the following holds:
$$S_3 \preccurlyeq^{\varepsilon_2}_u S_2 \text{ and } S_2 \preccurlyeq^{\varepsilon_1}_u S_1 \implies S_3 \preccurlyeq^{\varepsilon_1+\varepsilon_2}_u S_1.$$
\end{proposition}
\begin{proof}
Let $\mathcal{R}_2$ be the $\varepsilon_2$-ASUAS relation from $S_3$ to $S_2$ and $\mathcal{R}_1$ be the $\varepsilon_1$-ASUAS relation from $S_2$ to $S_1$. Consider the relation $\mathcal{R}\subseteq X_1 \times X_3$ defined by $(x_1,x_3) \in \mathcal{R}$ if and only if there exists $x_2 \in X_2$ such that $(x_1,x_2) \in \mathcal{R}_1$ and $(x_2,x_3) \in \mathcal{R}_2$, and let us show that $\mathcal{R}$ is an $(\varepsilon_1+\varepsilon_2)$-ASUAS relation from $S_3$ to $S_1$.

For $x_3^o \in X_3^o$, we have the existence of $x_2^o \in X_2^o$ such that $(x_2^o,x_3^o) \in \mathcal{R}_2$. Moreover, for such $x_2^o \in X_2^o$, we have the existence of $x_1^o \in X_1^o$ such that $(x_1^o,x_2^o) \in \mathcal{R}_1$. Hence, condition (i) of Definition~\ref{Def:altsimu_up} is satisfied.

For $(x_1,x_3) \in \mathcal{R}$, we have the existence of $x_2 \in X_2$ such that $(x_1,x_2) \in \mathcal{R}_1$ and $(x_2,x_3) \in \mathcal{R}_2$. Hence, one gets that $H_1(x_1) \leq H_2(x_2)+\varepsilon_1$ and $H_2(x_2) \leq H_3(x_3)+\varepsilon_2$, which implies that, $H_1(x_1) \leq H_3(x_3)+\varepsilon_1+\varepsilon_2$ and condition (ii) of Definition~\ref{Def:altsimu_up} is satisfied.


For $(x_1,x_3) \in \mathcal{R}$, we have the existence of $x_2 \in X_2$ such that $(x_1,x_2) \in \mathcal{R}_1$ and $(x_2,x_3) \in \mathcal{R}_2$. From the relations $\mathcal{R}_2$ and $\mathcal{R}_1$ we have $U_3 \subseteq \{\uparrow U_2\}$ and $U_2 \subseteq \{\uparrow U_1\}$, which imply $U_3 \subseteq \{\uparrow U_1\}$. Now pick any $u_3 \in U^a_3(x_3)$ and any $u_1 \in \{\downarrow u_3\} \cap U_1$ and let us show that $u_1 \in U^a_1(x_1)$. Under the hypothesis $U_1 \subseteq U_2$ or $U_3 \subseteq U_2$, define $u_2 := u_1$ in the first case and $u_2 := u_3$ in the second. In either case, $u_2 \in U_2$ and $u_1 \leq u_2 \leq u_3$. Hence $u_2 \in \{\downarrow u_3\} \cap U_2$, and by condition (iii) of $\mathcal{R}_2$, $u_2 \in U^a_2(x_2)$. Applying condition (iii) of $\mathcal{R}_1$ with this $u_2$, we conclude that $u_1 \in \{\downarrow u_2\} \cap U_1 \subseteq U^a_1(x_1)$, and condition (iii) of Definition~\ref{Def:altsimu_up} is satisfied.

For $(x_1,x_3) \in \mathcal{R}$, we have the existence of $x_2 \in X_2$ such that $(x_1,x_2) \in \mathcal{R}_1$ and $(x_2,x_3) \in \mathcal{R}_2$. Pick any $u_3 \in U^a_3(x_3)$, any $u_1 \in U^a_1(x_1) \cap \{\downarrow u_3\}$, and any $x_1' \in \Delta_1(x_1,u_1)$. Under the hypothesis $U_1 \subseteq U_2$ or $U_3 \subseteq U_2$, define $u_2 := u_1$ in the first case and $u_2 := u_3$ in the second. In either case, $u_2 \in U_2$ and $u_1 \leq u_2 \leq u_3$, so $u_2 \in \{\downarrow u_3\} \cap U_2 \subseteq U^a_2(x_2)$ by condition (iii) of $\mathcal{R}_2$. Applying condition (iv) of $\mathcal{R}_1$ with $u_2 \in U^a_2(x_2)$ and $u_1 \in U^a_1(x_1) \cap \{\downarrow u_2\}$, there exists $x_2' \in \Delta_2(x_2,u_2)$ such that $(x_1',x_2') \in \mathcal{R}_1$. Applying condition (iv) of $\mathcal{R}_2$ with $u_3 \in U^a_3(x_3)$ and $u_2 \in U^a_2(x_2) \cap \{\downarrow u_3\}$, there exists $x_3' \in \Delta_3(x_3,u_3)$ such that $(x_2',x_3') \in \mathcal{R}_2$. By the definition of $\mathcal{R}$, $(x_1',x_3') \in \mathcal{R}$, and condition (iv) of Definition~\ref{Def:altsimu_up} is satisfied.
 
\end{proof}

\subsection{Refinement Procedure under ASUAS}

In this section, we show that the  relation $S_2 \preccurlyeq^{\varepsilon}_u S_1$ allows refining a controller for $S_2$ into a controller for $S_1$, under lower-closed specification. To do so, we first introduce the concept of controllers for transition systems.

A controller for a system $S_1:=(X_1,X_1^o,U_1,\Delta_1,Y_1,H_1)$ is a transition system $S_c:=(X_c,X^o_c,U_c,\Delta_c,Y_c,H_c)$ such that $U_c \subseteq U_1$, $Y_1$ and $Y_c$ are subsets of the same set $Y$, and $S_c \preccurlyeq^{\varepsilon}_u S_1$ via a relation $\mathcal{R}_{1,c} \subseteq X_1 \times X_c$, for some $\varepsilon \in \mathbb{R}^n_+$. In this case, the controlled system is given by the transition system $S_{1c}:=S_c \times_{\mathcal{R}_{1,c}} S_{1}=(X_{1c},X^o_{1c},U_{1c},\Delta_{1c},Y_{1c},H_{1c})$ defined formally as follows: 
   \begin{itemize}
       \item The set of states $X_{1c}= (x_1,x_c) \in \mathcal{R}_{1,c}$;
       \item The set of initial states $X_{1c}^0=X_{1c} \cap (X_1^o \times X_c^o)$;
       \item The set of inputs $U_{1c}=U_c$;
       \item The transition relation $(x_1',x_c') \in \Delta_{1c}(x_1,x_c,u)$ holds if and only if 
       \begin{itemize}
           \item $(x_1,x_c)\in \mathcal{R}_{1,c}$;
           \item $u \in U^a_c(x_c)$;
           \item $x_1' \in \Delta_1(x_1,u_1)$, for some  $u_1 \in  \{\downarrow u\} \cap U^a_1(x_1)$;
           \item $x_c' \in \Delta_c(x_c,u)$ such that $(x_1',x_c') \in \mathcal{R}_{1,c}$;
       \end{itemize}
       \item The set of outputs $Y_{1c}=Y_1$;
       \item The output map $H_{1c}(x_1,x_c)=H_1(x_1)$.
   \end{itemize}

   \begin{remark} \label{remcl}
    Given $(x_1,x_c) \in X_{1c}$, the set of enabled inputs for the controlled system is given by $U^a_{1c}(x_1,x_c)=U_c^a(x_c)$. In this context, the controller works as follows: starting from a pair $(x_1, x_c) \in \mathcal{R}_{1,c}$, the controller $S_c$ allows the execution of any input  $u \in U_c^a(x_c)$. The system $S_1$ then selects an input $u_1 \in U_1^a(x_1) \cap  \{\downarrow u\}$ and transitions to any state $x_1' \in \Delta_1(x_1, u_1)$ under the chosen input. This transition is then matched by a corresponding transition by the controller, where the controller $S_c$ observes the new state $x_1'$ of $S_1$ and transitions to a state $x_c' \in \Delta_c(x_c, u)$ satisfying $(x_1', x_c') \in \mathcal{R}_{1,c}$. The existence of this matching transition is ensured by the fact that $S_c \preccurlyeq^{\varepsilon}_u S_1$. 
   \end{remark} 
   
In the following, we show how to refine a controller $S_{C_2}$ for system $S_2$, under a lower-closed  specification, into a controller $S_{C_1}$ for the transition system $S_1$, when $S_2 \preccurlyeq^{\varepsilon}_u S_1$. To do so, we start introducing an auxiliary behavioral relation.
\begin{definition}\label{Def:simu_up}
Consider the systems $S_i:=(X_i,X_i^o,U_i,\Delta_i,Y_i,H_i)$, $i=1,2$. Under Assumption \ref{assPO}, and for some $\varepsilon \in \mathbb{R}^m_{+}$, a relation $\mathcal{R} \subseteq X_1 \times X_2$ is said to be an $\varepsilon$-approximate simulation relation, or $\mathcal{R}$ is an $\varepsilon$-AUS, from $S_1$ to $S_2$ if: 
\begin{itemize}
    \item[(i)] $\forall x_1^o \in X_1^o$, $\exists x_2^o \in X_2^o$ such that $(x_1^o, x_2^o) \in \mathcal{R}$;
    \item[(ii)] $\forall (x_1,x_2) \in \mathcal{R}$, $H_1(x_1)\leq H_2(x_2)+\varepsilon$;
    \item[(iii)] $\forall (x_1,x_2) \in \mathcal{R}$, $\forall u_1 \in U^a_1(x_1)$ and $\forall x_1'\in \Delta_1(x_1,u_1)$, $\exists u_2 \in U^a_2(x_2)$ and $\exists x_2'\in \Delta_2(x_2,u_2)$ such that $(x_1',x_2') \in \mathcal{R}$.
\end{itemize}
In this case, we say that $S_2$ $\varepsilon$-approximately upper simulates $S_1$, or $S_2$ $\varepsilon$-AUS $S_1$, denoted $S_1 \preccurlyeq^{\varepsilon}_{u,S} S_2$.
\end{definition}

Unlike the traditional approximate simulation relation in~\cite{girard2007approximation}, which imposes output closeness, the proposed AUS relation imposes output ordering. Note that this relation is introduced solely as a proof tool for Theorem~\ref{thm:1}. The central behavioral relation throughout the paper remains the ASUAS relation of Definition~\ref{Def:altsimu_up}.

\begin{theorem}
\label{thm:1}
Consider the transition systems $S_i:=(X_i,X_i^o,U_i,\Delta_i,Y_i,H_i)$, $i=1,2$, such that Assumption~\ref{assPO} holds. Suppose that the relation $\mathcal{R}_{1,2}$ is an $\varepsilon$-ASUAS relation from $S_2$ to $S_1$, for some $\varepsilon \in \mathbb{R}^m_{+}$. Let $S_c$ be a controller for $S_2$ under a lower-closed specification $\phi$, i.e., $\mathcal{B}(S_c \times_{\mathcal{R}_{2,c}} S_2) \subseteq \phi$, where $\mathcal{R}_{2,c}$ is the corresponding ASUAS relation from $S_c$ to $S_2$. If $U_c \subseteq U_1$, then $S_c \times_{\mathcal{R}_{2,c}} S_2$ is a controller for the transition system $S_1$ under the specification $\Omega_{\varepsilon}(\phi)$.
\end{theorem}
\begin{proof}
Consider the relation 
\begin{eqnarray*}
    \mathcal{R} := & \{((x_1,(x_2,x_c)) \in X_1 \times (X_2 \times X_c) : 
    \\ & (x_2,x_c) \in \mathcal{R}_{2,c} ~ \land ~ (x_1,x_2) \in \mathcal{R}_{1,2}\}
\end{eqnarray*}
To prove the result, let us first claim that $\mathcal{R}$ is $\varepsilon$-ASUAS relation from $(S_c \times_{\mathcal{R}_{2,c}} S_{2}):= (X_{2c},X_{2c}^o,U_{2c},\Delta_{2c},Y_{2c},H_{2c})
$ to $S_1$. The latter would imply, using Proposition \ref{prop:4} in the Appendix, that  
$$ \left( (S_c \times_{\mathcal{R}_{2,c}} S_{2}) \times_{\mathcal{R}} S_1 \right) \preccurlyeq^{\varepsilon}_{u,S} \left( S_c \times_{\mathcal{R}_{2,c}} S_{2} \right), $$
where $\preccurlyeq^{\varepsilon}_{u,S}$ denotes the $\varepsilon$-AUS relation in Definition~\ref{Def:simu_up}. Then, using Proposition \ref{prop:5} in the Appendix, the fact that $\mathcal{B}(S_c \times_{\mathcal{R}_{2,c}}  S_2) \subseteq \phi$, and lower closedness of the specification $\phi$, that $S_c \times_{\mathcal{R}_{2,c}}  S_2$ is a controller for the transition system $S_1$ under the specification $\Omega_{\varepsilon}(\phi)$.

To prove the claim, we start using Assumption \ref{assPO} to conclude that $U_{2c} = U_c \subseteq U_1$ and $Y_{2c} = Y_2$ and $Y_1$ are subsets of the same set $Y \subseteq \mathbb{R}^m$.
For $(x_2^0,x_c^0) \in X_{2c}^0=\mathcal{R}_{2,c}\cap (X_2^o \times X_c^o) $, one has from definition of the controlled system $S_c \times_{\mathcal{R}_{2,c}} S_{2}$ that  
\begin{equation}
\label{eqn1}
    (x_2^0,x_c^0) \in \mathcal{R}_{2,c}.
\end{equation} 
Moreover,  one has that for $x_2^0 \in X_2^0$, there exists $x_1^0 \in X_1^0$ such that $(x_1^0,x_2^0) \in \mathcal{R}_{1,2}$. Hence, it follows from (\ref{eqn1}) that, for $(x_2^0,x_c^0) \in X_{2c}^0$, there exists $x_1^0 \in X_1^0$ such that $(x_1^0,(x_2^0,x_c^0)) \in \mathcal{R}$ and condition (i) of Definition \ref{Def:altsimu_up} is satisfied.
Now, consider $(x_1,(x_2,x_c)) \in \mathcal{R}$, we have that $ (x_1,x_2) \in \mathcal{R}_{1,2}$, which implies that $H_1(x_1) \leq H_2(x_2)+\varepsilon$. Hence, one gets that $H_1(x_1) \leq H_{2c}(x_2,x_c)+\varepsilon=H_2(x_2)+\varepsilon$ and condition (ii) of Definition \ref{Def:altsimu_up} is satisfied. 
To show the third condition, since $S_c$ is a controller for $S_2$, $U_c \subseteq U_2$, which combined with $U_2 \subseteq \{\uparrow U_1\}$ from $\mathcal{R}_{1,2}$ yields $U_{2c} = U_c \subseteq U_2 \subseteq \{\uparrow U_1\}$. Now consider $(x_1,(x_2,x_c)) \in \mathcal{R}$ and any $u \in U^a_{2c}(x_2,x_c) = U^a_c(x_c)$. Since $U_c \subseteq U_2$, we have $u \in U_2$, and from condition (iii) of $\mathcal{R}_{2,c}$ it follows that $\{\downarrow u\} \cap U_2 \subseteq U^a_2(x_2)$. In particular, $u \in U^a_2(x_2)$. Applying condition (iii) of $\mathcal{R}_{1,2}$ with $u$ as the abstract input, $\{\downarrow u\} \cap U_1 \subseteq U^a_1(x_1)$. Hence condition (iii) of Definition~\ref{Def:altsimu_up} is satisfied.
To show (iv), we consider $(x_1,(x_2,x_c)) \in \mathcal{R}$, any $u \in U^a_{2c}(x_2,x_c) = U^a_c(x_c)$, any $u_1 \in U^a_1(x_1) \cap \{\downarrow u\}$, and any $x_1' \in \Delta_1(x_1,u_1)$ and let us show the existence of $(x_2',x_c') \in \Delta_{2c}(x_2,x_c,u)$ such that $(x_1',(x_2',x_c')) \in \mathcal{R}$. Since $U_c \subseteq U_2$ (by the controller definition), $u \in U_2$, and from condition (iii) of $\mathcal{R}_{2,c}$, $\{\downarrow u\} \cap U_2 \subseteq U^a_2(x_2)$, in particular, $u \in U^a_2(x_2)$. Applying condition (iv) of $\mathcal{R}_{1,2}$ with $u \in U^a_2(x_2)$ and $u_1 \in U^a_1(x_1) \cap \{\downarrow u\}$, there exists $x_2' \in \Delta_2(x_2,u)$ such that $(x_1',x_2') \in \mathcal{R}_{1,2}$. Applying condition (iv) of $\mathcal{R}_{2,c}$ with $u \in U^a_c(x_c)$ and $u \in U^a_2(x_2) \cap \{\downarrow u\}$, there exists $x_c' \in \Delta_c(x_c,u)$ such that $(x_2',x_c') \in \mathcal{R}_{2,c}$. By the construction of the controlled system $S_c \times_{\mathcal{R}_{2,c}} S_2$, we have $(x_2',x_c') \in \Delta_{2c}(x_2,x_c,u)$. Combined with $(x_1',x_2') \in \mathcal{R}_{1,2}$, this gives $(x_1',(x_2',x_c')) \in \mathcal{R}$. 
\end{proof}

\section{Complete Model-based Abstractions} \label{sec:4}

\subsection{Sparse Abstractions for ASUAS}
    
We start constructing upper- and lower-sparse abstractions for the discrete-time monotone control system $\Sigma$ in (\ref{dis_sys}). 

An upper-sparse abstraction 
\begin{equation}
\label{eqn:up_abs}
 \mathcal{U}_\Sigma := \left(X_\mathcal{U}, X_{\mathcal{U}}^o, U_\mathcal{U}, \Delta_\mathcal{U}, Y_\mathcal{U}, \\H_\mathcal{U} \right) 
 \end{equation}
of $\Sigma$ is constructed as follows:

\begin{itemize}
    \item The set of states $X_\mathcal{U}$ is constructed based on a discretization of the state-space into $n_x\ge 1$ states using a finite partition\footnote{In defining a partition, we ignore the set of measure zero where intervals overlap for notational convenience, as is done for example in~\cite{coogan2015efficient}.} $X_\mathcal{U}$ of the set $X$. Each element $q$ of the partition can be described as an interval $q=[x_1^q,x_2^q]$; 
    \item The set of initial states $X_{\mathcal{U}}^o =\{q \in X_{\mathcal{U}} \mid q\cap X^o \neq \emptyset \}$; 
    \item The finite set of inputs $U_\mathcal{U}$ is constructed by approximating $U$ with $n_u$ values, $U_\mathcal{U}=\big\{\mathsf{u}_{\ell}\in U|\; \ell=0,\dots,n_u-1 \big\}$, with $\mathsf{u}_{0}=\underline{U}$ and $\mathsf{u}_{n_u-1}=\overline{U}$;
    \item The set of outputs is given by $Y_\mathcal{U}=X$;
    \item The output map is given for $q \in X_\mathcal{U}$ by $H_\mathcal{U}(q)=x_2^q$;
    \item The transition relation $\Delta_\mathcal{U} \subseteq X_\mathcal{U}\times U_\mathcal{U} \times X_\mathcal{U}$, defined  for $q, q'\in X_\mathcal{U}$ and $u\in U_\mathcal{U}$ as follows: $q' \in \Delta_\mathcal{U}(q,u)$ if and only if there exists $m\in \{1,\ldots,M\}$ such that $f(x_2^q,u, d_2^m) \in [x_1^{q'},x_2^{q'}]$.
\end{itemize}

Similarly, a lower-sparse abstraction 
\begin{equation}
    \label{eqn:low_abs}
    \mathcal{L}_\Sigma :=\left(X_{\mathcal{L}},X^o_{\mathcal{L}},U_{\mathcal{L}},\Delta_{\mathcal{L}},Y_{\mathcal{L}},H_{\mathcal{L}} \right)
\end{equation}
  of $\Sigma$ is constructed analogously to the upper-sparse abstraction $\mathcal{U}_\Sigma$ in \eqref{eqn:up_abs}: the set of states $X_{\mathcal{L}}$, the set of initial states $X^o_{\mathcal{L}}$, the set of inputs $U_{\mathcal{L}}$, and the output set $Y_{\mathcal{L}}=X$ are defined exactly as their counterparts in $\mathcal{U}_\Sigma$. The two abstractions differ only in the output map and the transition relation, described as follows:
\begin{itemize}
    \item The output map is given for $q \in X_{\mathcal{L}}$ by $H_{\mathcal{L}}(q)=x_1^q$;
    \item The transition relation $\Delta_\mathcal{L} \subseteq X_\mathcal{L}\times U_\mathcal{L} \times X_\mathcal{L}$ defined for $q,q'\in X_\mathcal{L}$ and $u\in U_\mathcal{L}$ as follows $q' \in \Delta_\mathcal{L}(q,u)$ if and only if there exists $m\in \{1,\ldots,M\}$ such that $f(x_1^q,u, d_2^m) \in [x_1^{q'},x_2^{q'}]$.
\end{itemize}

Next, we make the following assumption relating the sets $X_{\mathcal{U}}$ and $X_{\mathcal{L}}$ in $\mathcal{U}_\Sigma$ and $\mathcal{L}_\Sigma$, respectively.

\begin{assumption}
\label{ass33}
The lower and upper sparse abstractions $\mathcal{U}_{\Sigma}$ and $\mathcal{L}_{\Sigma}$ defined in (\ref{eqn:up_abs}) and (\ref{eqn:low_abs}), respectively, satisfy $X_{\mathcal{U}}=X_{\mathcal{L}}$.
\end{assumption}

Since $X_\mathcal{U}=X_\mathcal{L}$, we define the quantizer $Q_{X_\mathcal{U}}=Q_{X_{\mathcal{L}}}: X \rightarrow X_\mathcal{U}$ associated to the partition $X_\mathcal{U}$ as follows: for $x \in X$ and $q \in X_\mathcal{U}$, $q \in Q_{X_\mathcal{U}}(x)$ if and only if $x\in q$. An illustration of the construction procedure of the proposed lower and upper sparse abstractions $\mathcal{L}_\Sigma$ and $\mathcal{U}_\Sigma$ is provided in Figure~\ref{fig:abst}. 
\begin{figure}[!t]
	\begin{center}
		\includegraphics[scale=0.65]{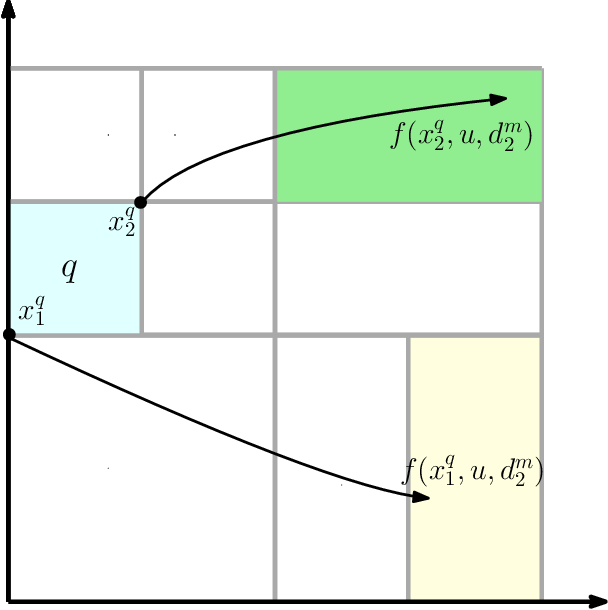}\\
	\end{center}
	\caption{A partition $X_{\mathcal{U}}=X_{\mathcal{L}}$ of the  state space, made of $9$ discrete states. The successor of the discrete state $q$ (in blue) for the upper (respectively, lower) sparse abstractions is the green state (respectively, the yellow state).}
	\label{fig:abst}	
\end{figure}

We are now ready to establish the ASUAS relations between the system $S_\Sigma$ and the proposed upper- and lower-sparse abstractions.

\begin{theorem} \label{thm:main}
Consider the discrete-time monotone control system $S_\Sigma=(X,X^o,U,\Delta,Y,H)$ in \eqref{eqSsigma} with $X\subseteq \mathbb{R}^n$ being lower closed, and such that Assumption \ref{assum:contr_dist} holds. Let $\mathcal{U}_\Sigma=\left(X_\mathcal{U}, X_{\mathcal{U}}^o, U_\mathcal{U}, \Delta_\mathcal{U}, Y_\mathcal{U}, \\H_\mathcal{U} \right)$ in \eqref{eqn:up_abs} and  $\mathcal{L}_\Sigma=\left(X_{\mathcal{L}},X^o_{\mathcal{L}},U_{\mathcal{L}},\Delta_{\mathcal{L}},Y_{\mathcal{L}},H_{\mathcal{L}} \right)$ in \eqref{eqn:low_abs} be an upper-sparse and 
a lower-sparse abstraction of $S_\Sigma$, respectively, satisfying Assumption \ref{ass33}. Then, 
\begin{equation}
\label{eqn:1}
\mathcal{U}_\Sigma  \preccurlyeq^0_u S_\Sigma \preccurlyeq^0_u \mathcal{L}_\Sigma.
\end{equation}
\end{theorem}

\begin{proof}
To prove the result, we show each of ASUAS relations separately.
\textbf{\underline{$S_{\Sigma} \preccurlyeq^0_u \mathcal{L}_{\Sigma}$}
:} Consider the relation $\mathcal{R} \subseteq X_{\mathcal{L}} \times X$ defined as follows:
\begin{equation}
\label{eqn:proof2}
    \mathcal{R} := \{(q,x) \in X_{\mathcal{L}} \times X \mid x_1^q \leq x \}.
\end{equation}
Let us show that $\mathcal{R}$ is a $0$-ASUAS relation from $S_{\Sigma}$ to $ \mathcal{L}_{\Sigma}$. Consider $x \in X^o$, from the construction of the set $X^o_{\mathcal{L}}$ we have the existence of $q \in X^o_{\mathcal{L}}$ such that $x_1^q \leq x$, which implies that $(q,x) \in \mathcal{R}$. Hence, the first condition of Definition~\ref{Def:altsimu_up} is directly satisfied. Now consider $(q,x) \in \mathcal{R}$, we have $H_{\mathcal{L}}(q)=x_1^q \leq x=H(x)$ and condition (ii) of Definition~\ref{Def:altsimu_up} is satisfied. 

For $(q,x) \in \mathcal{R}$ and $u \in U^a(x)$, we have that $\Delta(x,u) \subseteq X$. Consider $u_q \in \{\downarrow u\} \cap U_{\mathcal{L}}$, we have from the monotonicity of the system $\Sigma$ that $f(x_1^q,u_q, d_2^m) \in \{\downarrow \Delta(x,u)\} \subseteq \{\downarrow X\} \subseteq X$, where the last inclusion comes from the lower closedness of the set $X$. Hence, there exists $q' \in X_{\mathcal{L}}$ such that $f(x_1^q,u_q, d_2^m) \in q'$ and $q' \in \Delta_{\mathcal{L}}(q,u_q)$, which in turn implies that $u_q \in U^a_{\mathcal{U}}(q)$. Finally, since $\underline{U} \in U_{\mathcal{L}}$ by construction and $\underline{U} \leq u$ for all $u \in U$, we have $U \subseteq \{\uparrow U_{\mathcal{L}}\}$, and condition (iii) of Definition~\ref{Def:altsimu_up} is satisfied.

Now for $(q,x) \in \mathcal{R}$ and $u \in U^a(x)$, considering $u_q \in U^a_{\mathcal{L}}(q)$ such that $u_q \leq u$, the existence of such $u_q$ is guaranteed in view of Remark~\ref{rk:existence}. Now for $q' \in \Delta_{\mathcal{L}}(q,u_q)$, we have the existence of $m \in \{1,2,\ldots,M\}$ such that $f(x_1^q,u_q, d_2^m) \in [x_1^{q'},x_2^{q'}]$. Moreover, from the monotonicity of the system $\Sigma$ and since $x_1^q \leq x$ and $u_q \leq u$ we have the existence of $x' = f(x,u, d_2^m) \in \Delta(x,u)$ satisfying $x_1^{q'} \leq f(x_1^q,u_q, d_2^m)\leq f(x,u, d_{2}^m)=x'$. Then, $(q',x') \in \mathcal{R}$ and condition (iv) of Definition~\ref{Def:altsimu_up} is satisfied.

\textbf{\underline{$\mathcal{U}_{\Sigma}  \preccurlyeq^0_u S_{\Sigma}$}
:} Consider the relation $\mathcal{R} \subseteq X \times X_{\mathcal{U}}$ defined as follows:
\begin{equation}
\label{eqn:proof2_2}
    \mathcal{R}=\{(x,q) \in X \times X_{\mathcal{U}} \mid x \leq x_2^q \}
\end{equation}
Let us show that $\mathcal{R}$ is a $0$-ASUAS relation from $\mathcal{U}_{\Sigma}$ to $S_{\Sigma}$. For $q \in X^o_{\mathcal{U}}$. By definition of $X^o_{\mathcal{U}}$, there exists 
$x \in X^o \cap q$. Since $x \in q = [x_1^q, x_2^q]$, we have $x \leq x_2^q$, 
and hence $(x,q) \in \mathcal{R}$. Thus, condition (i) of 
Definition~\ref{Def:altsimu_up} is satisfied. Now for $(x,q) \in \mathcal{R}$, we have $H(x)=x \leq H_{\mathcal{U}}(q)=x_2^q $ and condition (ii) of Definition~\ref{Def:altsimu_up} is satisfied. 

For $(x,q) \in \mathcal{R}$ and 
$u_q \in U^a_{\mathcal{U}}(q)$, we have the existence of $q'\in X_{\mathcal{U}}$ such that  
$f(x_2^q,u_q,d_2^m) \in q' \subseteq X$. 
Now for 
$u \in \{\downarrow u_q\} \cap U$, we have from the monotonicity of the system $\Sigma$ that $x' \in \Delta(x,u) \in \downarrow f(x_2^q,u_q,d_2^m)  \subseteq \downarrow X \subseteq X$, where the last inclusion follows the lower closedness of the set $X$, and one gets $u \in U^a(x)$. Finally, since $U_{\mathcal{U}} \subseteq U$ by construction, we have $U_{\mathcal{U}} \subseteq \{\uparrow U\}$, and condition (iii) of Definition~\ref{Def:altsimu_up} is satisfied.

Now for $(x,q) \in \mathcal{R}$ and $u_q \in U^a_{\mathcal{U}}(q)$, considering $u \in U^a(x)$ such that $u \leq u_q$, the existence of such $u$ is guaranteed in view of Remark~\ref{rk:existence}. Now for $x' \in \Delta(x,u)$, we have the existence of $m \in \{1,2,\ldots,M\}$ and $d \in [d_1^m,d_2^m]$ such that $x'=f(x,u,d)$. Similarly, we have the existence of $q' \in \Delta (q,u_q)$ such that $f(x_2^q,u_q,d_2^m) \in q'=[x_1^{q'},x_2^{q'}]$. since $x \leq x_2^q$ and $u \leq u_q$ one gets from the monotonicity of the system $\Sigma$ that $x'=f(x,u, d) \leq f(x_2^q,u_q,d_2^m) \leq x_2^{q'}$. Then, $(x',q') \in \mathcal{R}$ and condition (iv) of Definition~\ref{Def:altsimu_up} is satisfied.
\end{proof}

According to the refinement procedure in Theorem~\ref{thm:1}, a controller synthesized for $\mathcal{U}_{\Sigma}$, under a lower-closed specification, can be refined into a controller for $S_{\Sigma}$. Conversely, the absence of controller for $\mathcal{L}_{\Sigma}$ certifies the absence of controller for $S_{\Sigma}$. We refer to the pair $(\mathcal{U}_{\Sigma},\mathcal{L}_{\Sigma})$ of Theorem~\ref{thm:main} as a \emph{complete abstraction pair} for $S_{\Sigma}$. This completeness property is of significant practical importance. While existing sound abstractions in the literature guarantee that a controller synthesized for the abstraction can be refined into one for the original 
system~\cite{tabuada2009verification,belta2017formal,girard2007approximation,hsu2018multi}, they are inherently one-sided: the absence of an abstract controller leaves the question of whether a controller for the original system exists entirely open. Our approach addresses this limitation by additionally leveraging the lower-sparse abstraction, which allows one to conclusively certify the existence or non-existence of a controller for the original system.

\subsection{Controlling the Conservativeness of the Sparse Abstractions}

Theorem~\ref{thm:main} establishes that the lower- and upper-sparse abstractions together form a complete abstraction pair for monotone systems. However, this result may be conservative in cases where $\mathcal{L}_{\Sigma}$ admits a controller while $\mathcal{U}_{\Sigma}$ does not. 
Motivated by~\cite{puri1995varepsilon}, in this section, we propose an approach to mitigate the degree of conservativeness between $\mathcal{U}_{\Sigma}$ and $\mathcal{L}_{\Sigma}$.

Before providing our result, we make an assumption on the  considered set of states $X_{\mathcal{U}}=X_{\mathcal{L}}$ for $\mathcal{U}_{\Sigma}$ and $\mathcal{L}_{\Sigma}$ 
 defined in (\ref{eqn:up_abs}) and (\ref{eqn:low_abs}), respectively, and their corresponding quantizers $Q_{X_{\mathcal{U}}}=Q_{X_{\mathcal{L}}}$.

\begin{figure}
\centering
    \includegraphics[scale=0.55]{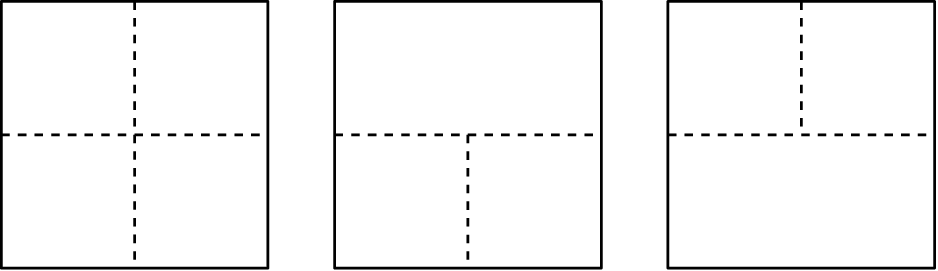}
\caption{The first two partitions satisfy Assumption~\ref{assum1}, while the third partition does not. }
\label{fig:partitions}
\end{figure}

\begin{assumption}
\label{assum1}
\begin{itemize}
    \item[(i)] For  $q,q' \in X_{\mathcal{U}}$ with $q=[x_1^q,x_2^q]$ and $q'=[x_1^{q'},x_2^{q'}]$, if $x_2^q \leq x_2^{q'}$ then $x_1^q \leq x_1^{q'}$.
    \item[(ii)] For $x,x' \in X$ if $x \leq x'$, then for $ q=[x_1^q,x_2^q]= Q_{X_\mathcal{U}}(x)$ and $ q'=[x_1^{q'},x_2^{q'}]= Q_{X_\mathcal{U}}(x')$ we have $x_2^q \leq x_2^{q'}$.
\end{itemize}
\end{assumption}

We note that the practically-relevant case of uniform rectangular grids, where all the cells have equal side lengths along each axis, always satisfies Assumption~\ref{assum1}. Fig.~\ref{fig:partitions} shows other examples illustrating Assumption~\ref{assum1}. We also make the following assumption.

\begin{assumption}
\label{ass:bounds}
    The map $f$ is continuously differentiable and there exist $\alpha_{ij} \geq 0$, $i,j \in \{1,2,\ldots,n\}$, such that
    \begin{equation*}
        \label{eqn:bound_x}
        \begin{aligned}
            0 &\leq \frac{\partial f_i}{\partial x_j} \leq \alpha_{ij} \qquad \forall i,j \in \{1,2,\ldots,n\}.
        \end{aligned}
    \end{equation*}
\end{assumption}

Theorem~\ref{thm:main} guarantees completeness but leaves the gap between $\mathcal{U}_{\Sigma}$ and $\mathcal{L}_{\Sigma}$ uncontrolled. To bound this gap by the discretization parameter, we relate the two abstractions through a \emph{perturbed} version of $\Sigma$ and use Assumption~\ref{ass:bounds} to translate a prescribed precision $\varepsilon$ into an admissible cell size $\eta$. To this end, we first introduce a perturbed version of the system $\Sigma$ in \eqref{dis_sys}: for $\varepsilon \in \mathbb{R}^n_{+}$,
\begin{equation}
\label{dis_sys_pert}
\Sigma_\varepsilon : x(k+1) \in f(x(k),u(k),d(k))+\Omega_{\varepsilon}(0). 
\end{equation}

\begin{theorem} \label{thm:main2}
 Consider the monotone control system \( S_\Sigma \) in \eqref{eqSsigma}, where  \( X \) is lower closed,  \( f \) satisfies Assumption~\ref{ass:bounds}, and Assumption~\ref{assum:contr_dist} holds. Let $\mathcal{U}_\Sigma$ and $\mathcal{L}_\Sigma$ in \eqref{eqn:up_abs} and \eqref{eqn:low_abs}, respectively, be upper-  and lower-sparse abstractions of $S_\Sigma$ such that Assumptions \ref{ass33} and \ref{assum1} hold. Then, given $\varepsilon \in \mathbb{R}^n_+$,  and $\mathcal{L}_{\Sigma_{\varepsilon}}$\footnote{For $\mathcal{L}_{\Sigma_{\varepsilon}}$, the set of states $X_{\mathcal{L}}$, initial states $X_{\mathcal{L}}^o$, control inputs $U_{\mathcal{L}}$, outputs $Y_{\mathcal{L}}$, and the output map $H_{\mathcal{L}}$ are inherited from the  $\mathcal{L}_{\Sigma}$, and the transition relation $\Delta_{\mathcal{L},\mathcal{\varepsilon}}$ is constructed according the dynamics $x(k+1) \in f(x(k),u(k),d(k))+\Omega_{\varepsilon}(0)$ of $S_{\Sigma_{\varepsilon}}$.} (a lower-sparse abstraction of $S_{\Sigma_{\varepsilon}}$), we have,
\begin{equation}
\label{eqn:3}
    \mathcal{L}_{\Sigma_{\varepsilon}}  \preccurlyeq^\eta_u \mathcal{U}_{\Sigma} \preccurlyeq^0_u \mathcal{L}_{\Sigma}
\end{equation}
provided that 
\begin{equation}
\label{eqPrecision}
\begin{aligned}
 \varepsilon_i &\geq \sum\limits_{j=1}^n\alpha_{ij}\eta_j \qquad \forall i \in \{1,2,...,n\},
 \\
\eta_i & =\max\{|x_{2,i}^q-x_{1,i}^q|:q \in X_{\mathcal{U}}\} ~~~ \forall i \in \{1,\ldots,n\}.
\end{aligned}
\end{equation}
\end{theorem}

\begin{proof}
By Theorem~\ref{thm:main}, $\mathcal{U}_{\Sigma} \preccurlyeq^0_u S_{\Sigma}$ and $S_{\Sigma} \preccurlyeq^0_u \mathcal{L}_{\Sigma}$. Identifying $S_3=\mathcal{U}_{\Sigma}$, $S_2=S_{\Sigma}$, $S_1=\mathcal{L}_{\Sigma}$ in Proposition~\ref{prop:transitivity}, the hypothesis $U_3 \subseteq U_2$ holds since $U_{\mathcal{U}} \subseteq U$ by the construction of the upper-sparse abstraction. Applying Proposition~\ref{prop:transitivity}, one gets $\mathcal{U}_{\Sigma} \preccurlyeq^0_u \mathcal{L}_{\Sigma}$.

Lets us show that $\mathcal{L}_{\Sigma_{\varepsilon}}  \preccurlyeq^\eta_u \mathcal{U}_{\Sigma}$. Consider the relation $\mathcal{R} \subseteq X_{\mathcal{U}} \times X_{\mathcal{L}}$ defined by $(q_1,q_2) \in \mathcal{R}$ if and only if $x_2^{q_1} \leq x_2^{q_2}$. Let us show that $\mathcal{R}$ is an $\eta$-ASUAS relation from $\mathcal{L}_{\Sigma_{\varepsilon}}$ to $\mathcal{U}_{\Sigma}$.

Since $X^o_\mathcal{U}=X^o_\mathcal{L}$ condition (i) of Definition~\ref{Def:altsimu_up} is directly satisfied. Now consider $(q_1,q_2) \in \mathcal{R}$, we have $x_2^{q_1}\leq x_2^{q_2}$ and $H_{\mathcal{U}}(q_1)=x_2^{q_1} \leq x_2^{q_2} \leq H_{\mathcal{L}}(q_2)+\eta= x_1^{q_2}+\eta$, where the second inequality follows from (\ref{eqPrecision}). Hence, condition (ii) in Definition~\ref{Def:altsimu_up} is satisfied.

Consider $(q_1,q_2) \in \mathcal{R}$ and $u_{2} \in U^a_{\mathcal{L}}(q_2)$, hence we have that $\max\{f(x_1^{q_2},u_2,d_2^m)+\Omega_{\varepsilon}(0)\} \in X$. Now consider $u_1 \in \{\downarrow u_2\} \cap U_{\mathcal{U}}$. We have that
\begin{align}
\label{eqn:completeness}
\max\{f(x_1^{q_2},u_2,d_2^m)+&\Omega_{\varepsilon}(0)\} \geq  \max\{f(x_1^{q_1},u_2,d_2^m)+\Omega_{\varepsilon}(0)\} \nonumber  \\ &\geq \max\{f(x_2^{q_1}-\eta,u_1,d_2^m)+\Omega_{\varepsilon}(0)\} \nonumber\\ &\geq f(x_2^{q_1},u_1,d_2^m),
\end{align}
where the first inequality comes from (i) in Assumption~\ref{assum1}, the second inequality comes from the fact that $u_1 \leq u_2$ and $x_1^{q_1} \geq x_2^{q_1}-\eta$ in (\ref{eqPrecision}), and the last inequality follows from Lemma~\ref{lem:growth_bound} under Assumption \ref{ass:bounds}. Hence, one gets that $f(x_2^{q_1},u_1,d_2^m) \in \downarrow \max\{f(x_1^{q_2},u_2,d_2^m)+\Omega_{\varepsilon}(0)\} \subseteq \downarrow X \subseteq X$, where the last inclusion comes from the lower closedness of the set $X$, and one gets that $u_1 \in U^a_{\mathcal{U}}(q_1)$. Finally, since $U_{\mathcal{L}} = U_{\mathcal{U}}$ by Assumption~\ref{ass33}, we have $U_{\mathcal{L}} \subseteq \{\uparrow U_{\mathcal{U}}\}$, and condition (iii) of Definition~\ref{Def:altsimu_up} is satisfied.

Consider $(q_1,q_2) \in \mathcal{R}$, $u_{2} \in U^a_{\mathcal{L}}(q_2)$ and any $u_{1} \in U^a_{\mathcal{U}}$ satisfying $u_{1} \leq u_{2}$. For $q_1' \in \Delta_{\mathcal{U}}(q_1,u_1)$, we have the existence of $m \in \{1,2,\ldots, M\}$ such that $f(x_2^{q_1},u_1,d_2^m) \in q_1'$. Moreover, we have from (\ref{eqn:completeness}) that
$$ f(x_2^{q_1},u_1,d_2^m) \leq \max\{f(x_1^{q_2},u_2,d_2^m)+\Omega_{\varepsilon}(0)\}, $$ 
Hence, one gets from (ii) in Assumption~\ref{assum1} the existence of $$q_2' \in Q_{X_\mathcal{U}}(\max\{f(x_1^{q_2},u_2,\\d_2^m)+\Omega_{\varepsilon}(0)\})$$ 
such that $x_2^{q_1'} \leq x_2^{q_2'}$. Finally, since 
$$ q_2' \in Q_{X_\mathcal{U}}(\max\{f(x_1^{q_2},\\u_2,d_2^m)+\Omega_{\varepsilon}(0)\}), $$ 
one gets that $q_2' \in \Delta_{\mathcal{L},\varepsilon}(q_2,u_2)$ satisfies $x_2^{q_1'} \leq x_2^{q_2'}$ which implies that $(q_1',q_2') \in \mathcal{R}$ and condition (iv) in Definition~\ref{Def:altsimu_up} is satisfied. 
\end{proof}

According to Theorem \ref{thm:main2},  to achieve a desired precision $\varepsilon$ between the lower- and the 
upper-sparse abstractions, the largest discretization parameter $\eta \in \mathbb{R}^n_{+}$ should satisfy  \eqref{eqPrecision}. This result is important in practice. Indeed, it shows that any desired accuracy can be achieved using a sufficiently fine  discretization. More importantly, it indicates, for a desired precision $\varepsilon$, the maximum value of the discretization parameter $\eta$ to  use.
Combining Theorems~\ref{thm:main} and~\ref{thm:main2}, and Proposition~\ref{prop:transitivity}, we can formulate the following corollary showing that $S_{\Sigma_\varepsilon}$ $\eta$-ASUAS $\mathcal{U}_{\Sigma}$.  

\begin{corollary}
\label{coro}
Under the preliminaries of Theorem~\ref{thm:main2},
\begin{equation}
\label{eqn:corro1}
    S_{\Sigma_{\varepsilon}} \preccurlyeq^\eta_u \mathcal{U}_{\Sigma} \preccurlyeq^0_u S_{\Sigma}.
\end{equation}
\end{corollary}
\begin{proof}
By combining Theorems~\ref{thm:main} and~\ref{thm:main2} one gets 
$ S_{\Sigma_{\varepsilon}}  \preccurlyeq^0_u \mathcal{L}_{\Sigma_{\varepsilon}}  \preccurlyeq^\eta_u \mathcal{U}_{\Sigma} \preccurlyeq^0_u S_{\Sigma}$. 
Identifying $S_3=S_{\Sigma_{\varepsilon}}$, $S_2=\mathcal{L}_{\Sigma_{\varepsilon}}$, $S_1=\mathcal{U}_{\Sigma}$ in Proposition~\ref{prop:transitivity}, the hypothesis $U_1 \subseteq U_2$ holds since $U_{\mathcal{U}} = U_{\mathcal{L}}$ by the construction of the two sparse abstractions. Applying Proposition~\ref{prop:transitivity} we obtain (\ref{eqn:corro1}).
\end{proof}

According to the refinement procedure in Theorem~\ref{thm:1}, the previous corollary guarantees approximate completeness of the upper-sparse abstraction, under a lower-closed specification $\phi$. That is, if we can synthesize a controller for  $\mathcal{U}_{\Sigma}$, then we can refine it into a controller for $S_{\Sigma}$, and if there is no controller for $\mathcal{U}_{\Sigma}$, then there is no controller for the system $S_{\Sigma_{\varepsilon}}$ under the specification $\Omega_{-\eta}(\phi)$\footnote{For a specification $\phi \subseteq Y^w$ and $\eta \in \mathbb{R}^n_{+}$, the set $\Omega_{-\eta}(\phi)$ is defined by $\Omega_{-\eta}(\phi):=\{\sigma_y=y_0,y_1,y_2,\ldots \in Y^w \mid \Omega_{\eta}(\sigma_y) \subseteq \phi$\}.}.
Moreover, any precision $\varepsilon \in \mathbb{R}^n_{+}$ can be achieved by choosing a discretization parameter $\eta \in \mathbb{R}^n_{+}$ according to (\ref{eqPrecision}).

\section{Complete Data-Driven Abstractions}
\label{sec:dd}
In this section, we propose a data-driven approach to construct the upper- and lower-sparse abstractions.
Instead of relying on the system's dynamics, we use a set of data points sampled from the system to construct the abstractions.
Data-driven abstractions are constructed on a partition similar to the one used in model-based abstractions. This will allow us to provide relationships between the data-driven and the model-based abstractions. 

\subsection{Data-Driven Abstractions for ASUAS}
\label{sec:dd_max}

 We first consider the case where all data points are collected under the maximum disturbance input. We then show how the approach can be extended to a more general setting, in which, data are collected under varying disturbance inputs, assuming that the Lipschitz constant of the map f with respect to the disturbance is known.

Let us consider a data set sampled from the discrete-time control system in \eqref{dis_sys}, while applying the maximum (worst-case) disturbance at each sample point. 
To simplify the results, we assume in this section that the set of disturbance inputs is a multidimentional interval of the form $D = [\underline D, \overline{D}] \subseteq \mathbb{R}^p$. 
The data set is defined as follows:
$$
\mathcal{D} = \{ (\tilde{x}_k,\tilde{u}_k, \tilde{x}_k') \mid \tilde{x}_k' = f(\tilde{x}_k,\tilde{u}_k,\overline{D}), k \in \mathbb K \},
$$
where $\mathbb{K}$ is a finite set of indices. For all $k \in \mathbb{K}$ we have $\tilde{x}_k \in X \subseteq \mathbb{R}^n$, $\tilde{u} \in U \subseteq \mathbb{R}^m$. 
We define the data-driven upper-sparse abstraction on the same partition used to define the model-based one. 
This abstraction is formally defined as  
\begin{equation}
\label{eqn:up_abs_data}
 \mathcal{U}_{\Sigma_\mathcal{D}} := \left(X_{\mathcal{U}_\mathcal{D}}, X_{\mathcal{U}_\mathcal{D}}^o, U_{\mathcal{U}_\mathcal{D}}, \Delta_{\mathcal{U}_\mathcal{D}}, Y_{\mathcal{U}_\mathcal{D}}, \\H_{\mathcal{U}_\mathcal{D}} \right), 
 \end{equation}
 where, $X_{\mathcal{U}_\mathcal{D}} = X_{\mathcal{U}},\ X_{\mathcal{U}_\mathcal{D}}^o = X_{\mathcal{U}}^o, U_{\mathcal{U}_\mathcal{D}} = U_{\mathcal{U}},  Y_{\mathcal{U}_\mathcal{D}} = Y_{\mathcal{U}}, \\H_{\mathcal{U}_\mathcal{D}} = H_{\mathcal{U}}$ are inherited from $\mathcal{U}_{\Sigma}$ in (\ref{eqn:up_abs}).

 To define the transition relation for the upper-sparse data-driven abstraction $\Delta_{\mathcal{U}_\mathcal{D}}$, we first introduce some auxiliary notations. Given an abstract state $q \in X_{\mathcal{U}_\mathcal{D}}$ and an abstract input $u \in U_{\mathcal{U}_\mathcal{D}}$, the map $(q,u) \mapsto \mathbb K^+(q,u) \subseteq \mathbb{K}$ defined as $\mathbb{K}^+(q,u) = \{k \in \mathbb{K} \mid x_2^q \leq \tilde{x}_k, u \leq \tilde{u}_k\}$ returns the collection of data points with larger state and input compared to $(q,u)$. Under the monotonicity assumption, the successor of the abstract state $q$ under input $u$ can be bounded above by the successors of data points that are greater than the pair $(q, u)$. This is captured by the following set:
\begin{equation}
\label{eqn:Q_u}
    \mathcal Q_u(q,u) = \left( \displaystyle \bigcap_{k \in \mathbb{K}^+(q,u)}\{x \in X \mid x \leq \tilde{x}_k' \} \right).
\end{equation}
Note that $\mathcal Q_u$ represents the region of the state space bounded above by the {transitions of data points } satisfying $x_2^q \leq \tilde{x}_k$ and $u \leq \tilde{u}_k$. 
Here, we adopt the convention that the intersection over an empty index set (i.e., when $\mathbb {K}^+(q,u) = \emptyset$) yields the entire state space $\mathcal Q_u(q,u) = X$.
Now we can define the transition relation as follows 
\begin{equation}
\label{eqn:delta_data}
 \Delta_{\mathcal{U}_\mathcal{D}}(q,u) = q' \Leftrightarrow x_2^{q'} = \max \{x_2^{q^\star} \mid {q^\star} \cap \mathcal Q_u(q,u) \neq \emptyset\}. 
 \end{equation}
The transition relation $ \Delta_{\mathcal{U}_\mathcal{D}}$ can be seen as choosing the largest abstract state intersecting with the allowed successor (computed from data). The unicity of the successor is due to the interval shape of the disturbance set. \\
Similarly, we define the data-driven lower-sparse abstraction. 
\begin{equation}
\label{eqn:low_abs_data}
 \mathcal{L}_{\Sigma_\mathcal{D}} := \left(X_{\mathcal{L}_\mathcal{D}}, X_{\mathcal{L}_\mathcal{D}}^o, U_{\mathcal{L}_\mathcal{D}}, \Delta_{\mathcal{L}_\mathcal{D}}, Y_{\mathcal{L}_\mathcal{D}}, \\H_{\mathcal{L}_\mathcal{D}} \right), 
 \end{equation}
 where $X_{\mathcal{L}_\mathcal{D}} = X_{\mathcal{L}},\ X_{\mathcal{L}_\mathcal{D}}^o = X_{\mathcal{L}}^o, U_{\mathcal{L}_\mathcal{D}} = U_{\mathcal{L}},  Y_{\mathcal{L}_\mathcal{D}} = Y_{\mathcal{L}}, \\H_{\mathcal{L}_\mathcal{D}} = H_{\mathcal{L}}$ are inherited from $\mathcal{L}_{\Sigma}$ in (\ref{eqn:low_abs}). Similarly to the data-driven upper sparse abstraction, to define the transiton relation $\Delta_{\mathcal{L}_\mathcal{D}}$ we consider the map $(q,u) \mapsto \mathbb K^-(q,u) \subseteq \mathbb{K}$ defined formally as
 {$\mathbb{K}^-(q,u) = \{k \in \mathbb{K} \mid x_1^q \geq \tilde{x}_k, u \geq \tilde{u}_k\}$}. Indeed, given an abstract state $q$ and an abstract input $u$, the set $\mathbb{K}^-(q,u)$ returns all data points smaller than the pair $(q,u)$. We bound the successor of the abstract state $q$ under an input $u$ from below by the successor of the data points smaller than the pair $(q,u)$ using
\begin{equation}
\label{eqn:Q_l}
 \mathcal Q_l(q,u) = \left(\displaystyle \bigcap_{k \in \mathbb {K}^-(q,u)}\{x \in X \mid x \geq \tilde{x}_k' \}\right).  
\end{equation}
The transition relation $\Delta_{\mathcal{L}_\mathcal{D}}$ is then defined as follows:  \begin{equation}
\label{eqn:delta_down_data}
 \Delta_{\mathcal{L}_\mathcal{D}}(q,u) =  q'  \Leftrightarrow x_1^{q'} = \min \{x_1^{q^\star} \mid {{q^\star} \cap \mathcal Q_l(q,u)} \neq \emptyset \}. 
 \end{equation}

 With the necessary preliminaries covered, we now  provide the relationship ordering between the original control system $S_{\Sigma}$, the sparse model-based abstractions $\mathcal{L}_{\Sigma},\mathcal{U}_{\Sigma}$ and the sparse data-driven abstractions $\mathcal{L}_{\Sigma_\mathcal{D}}, \mathcal{U}_{\Sigma_\mathcal{D}}$.

\begin{theorem} \label{thm:main_data}
Consider the  monotone control system $S_\Sigma=(X,X^o,U,\Delta,Y,H)$ in \eqref{eqSsigma} with $X\subseteq \mathbb{R}^n$ being lower closed, and such that Assumption \ref{assum:contr_dist} holds. Let $\mathcal{U}_\Sigma=\left(X_\mathcal{U}, X_{\mathcal{U}}^o, U_\mathcal{U}, \Delta_\mathcal{U}, Y_\mathcal{U}, \\H_\mathcal{U} \right)$ in \eqref{eqn:up_abs} and  $\mathcal{L}_\Sigma=\left(X_{\mathcal{L}},X^o_{\mathcal{L}},U_{\mathcal{L}},\Delta_{\mathcal{L}},Y_{\mathcal{L}},H_{\mathcal{L}} \right)$ in \eqref{eqn:low_abs} be an upper-sparse and 
a lower-sparse abstraction of $S_\Sigma$, respectively, satisfying Assumption \ref{ass33}. Similarly, let $\mathcal{U}_{\Sigma_\mathcal{D}} = \left(X_{\mathcal{U}_\mathcal{D}}, X_{\mathcal{U}_\mathcal{D}}^o, U_{\mathcal{U}_\mathcal{D}}, \Delta_{\mathcal{U}_\mathcal{D}}, Y_{\mathcal{U}_\mathcal{D}}, \\H_{\mathcal{U}_\mathcal{D}} \right)$ in \eqref{eqn:up_abs_data} and $\mathcal{L}_{\Sigma_\mathcal{D}} = \left(X_{\mathcal{L}_\mathcal{D}}, X_{\mathcal{L}_\mathcal{D}}^o, U_{\mathcal{L}_\mathcal{D}}, \Delta_{\mathcal{L}_\mathcal{D}}, Y_{\mathcal{L}_\mathcal{D}}, \\H_{\mathcal{L}_\mathcal{D}} \right)$ in \eqref{eqn:low_abs_data} be an upper-sparse and 
a lower-sparse data-driven abstraction of $S_\Sigma$, respectively, satisfying Assumption \ref{ass33}. Then, 
\begin{equation}
\label{eqn:soundness_model}
       \mathcal{U}_{\Sigma_\mathcal{D}} \preccurlyeq^0_u \mathcal{U}_\Sigma  \preccurlyeq^0_u S_\Sigma \preccurlyeq^0_u \mathcal{L}_\Sigma \preccurlyeq^0_u \mathcal{L}_{\Sigma_\mathcal{D}}.
\end{equation}
\end{theorem}
\begin{proof}
To prove the result, we show each of ASUAS relations separately.

\textbf{\underline{$\mathcal{L}_{\Sigma} \preccurlyeq^0_u \mathcal{L}_{\Sigma_\mathcal{D}}$}
:} Consider the relation $\mathcal{R} \subseteq X_{\mathcal{L}_{\mathcal{D}}} \times X_{\mathcal{L}}$ defined as follows:
\begin{equation}
\label{eqn:proof4_2}
    \mathcal{R} := \{(q_1,q_2) \in X_{\mathcal{L}_{\mathcal{D}}} \times X_{\mathcal{L}} \mid x_1^{q_1} \leq x_1^{q_2} \}.
\end{equation}
Let us show that $\mathcal{R}$ is a $0$-ASUAS relation from  $ \mathcal{L}_{\Sigma}$ to $\mathcal{L}_{\Sigma_\mathcal{D}}$. 
 Consider $q_2 \in X^o_{\mathcal{L}}$, and let us choose $q_1 \in X^o_{\mathcal{L}_{\mathcal{D}}}$ such that $x_1^{q_2} = x_1^{q_1}$. 
 This is possible because the two systems share the same set of initial states  $X^o_{\mathcal{L}}$. Hence, the first condition of Definition~\ref{Def:altsimu_up} is satisfied.
 Now consider $(q_1,q_2) \in \mathcal{R}$, we have $H_{\mathcal{L}_{\mathcal{D}}}(q_1)=x_1^{q_1} \leq x_1^{q_2}=H_{\mathcal{L}}(q_2)$ 
 and condition (ii) of Definition~\ref{Def:altsimu_up} is also satisfied. 

 Now, for $(q_1,q_2) \in \mathcal{R}$ and $u_{q_2} \in U^a_{\mathcal{L}}(q_2)$, we have that $\Delta_{\mathcal{L}}(q_2,u_{q_2}) \subseteq X$. 
For $u_{q_1} \in \{\downarrow u_{q_2}\} \cap U_{\mathcal{L}_\mathcal{D}}$, we have that $ x_1^{q_1'}= \min \{x_1^{q^\star} \mid {q^\star} \cap \mathcal Q_l(q_1,u_{q_1}) \neq \emptyset \}  \leq \min(\mathcal Q_l(q_1,u_{q_1})) \subseteq \{\downarrow X\} \subseteq X$,
 where the first inclusion comes from the definition of $\mathcal Q_l(q_1,u_{q_1})$ in (\ref{eqn:Q_l}) and the last inclusion comes from the lower closedness of the set $X$. Hence, there exists $q_{1}' \in X_{\mathcal{L}_\mathcal{D}}$ 
 such that $x_1^{q_1'} \in q_{1}'$ and $q_1' = \Delta_{\mathcal{L}_\mathcal{D}}(q_1,u_{q_1})$, which in turn implies that
  $u_{q_1} \in U^a_{\mathcal{L}_\mathcal{D}}(q_1)$. Finally, since $U_{\mathcal{L}_\mathcal{D}} = U_{\mathcal{L}}$ by construction, we have $U_{\mathcal{L}} \subseteq \{\uparrow U_{\mathcal{L}_\mathcal{D}}\}$, and condition (iii) of Definition~\ref{Def:altsimu_up} is satisfied.
 
For $(q_1,q_2) \in \mathcal{R}$ and $u_{q_2} \in U^a_{\mathcal{L}}(q_2)$, consider $u_{q_1} \in U^a_{\mathcal{L}_\mathcal{D}}(q_1)$ 
such that $u_{q_1} \leq u_{q_2}$, the existence of such $u_{q_1}$ is guaranteed in view of Remark~\ref{rk:existence}.
 Now, for $q_1' = \Delta_{\mathcal{L}_\mathcal{D}}(q_1,u_{q_1})$, we have $ x_1^{q_1'} = \min \{x_1^{q^\star} \mid {q^\star} \cap \mathcal Q_l(q_1,u_{q_1}) \neq \emptyset \} $, 
 making $x_1^{q_1'} \leq \min(\mathcal Q_l(q_1,u_{q_1}))$. We have $x_1^{q_2'}$ is defined such that $f(x_1^{q_2},u_{q_2}, \overline{D}) \in [x_1^{q_2'},x_2^{q_2'}]$,
meaning that $x_1^{q_2'} \leq f(x_1^{q_2},u_{q_2}, \overline{D})$ and $\forall q_2^{\star} \in  X_{\mathcal{L}}$ such that
$x_1^{q_2^{\star}} \leq f(x_1^{q_2},u_{q_2}, \overline{D})$, then $x_1^{q_2^{\star}} \leq x_1^{q_2'}$. 
Due to the monotonicity of the system, and since $x_1^{q_1} \leq x_1^{q_2} $ and $u_{q_1} \leq u_{q_2}$ 
we have for all $k \in \mathbb {K}^-(q_1,u_{q_1})$, $\tilde{x}_k' \leq f(x_1^{q_2},u_{q_2}, \overline{D})$.
Hence, $\min(\mathcal Q_l(q_1,u_{q_1})) \leq f(x_1^{q_2},u_{q_2}, \overline{D})$. 
Therefore, $x_1^{q_1'} \leq x_1^{q_2'}$, $(q_1',q_2') \in \mathcal{R}$ and condition (iv) of Definition~\ref{Def:altsimu_up} is satisfied.

\textbf{\underline{$\mathcal{U}_{\Sigma_\mathcal{D}}  \preccurlyeq^0_u \mathcal{U}_{\Sigma}$}
:} Consider the relation $\mathcal{R} \subseteq X_{\mathcal{U}} \times X_{\mathcal{U}_\mathcal{D}}$ defined as follows:
\begin{equation}
\label{eqn:proof4}
    \mathcal{R}=\{(q_1,q_2) \in X_{\mathcal{U}} \times X_{\mathcal{U}_\mathcal{D}} \mid x_2^{q_1} \leq x_2^{q_2} \}.
\end{equation}
Let us show that $\mathcal{R}$ is a $0$-ASUAS relation from  $ \mathcal{U}_{\Sigma_\mathcal{D}}$ to $\mathcal{U}_{\Sigma}$. 
 Consider $q_2 \in X^o_{\mathcal{U}_{\mathcal{D}}}$, and let us choose $q_1 \in X^o_{\mathcal{U}}$ such that $x_2^{q_2} = x_2^{q_1}$. 
 This is possible because the two systems share the same set of initial states $X^o_{\mathcal{U}}$. Hence, the first condition of Definition~\ref{Def:altsimu_up} is satisfied.
Let $(q_1,q_2) \in \mathcal{R}$, we have $H_{\mathcal{U}}(q_1)=x_2^{q_1} \leq x_2^{q_2}=H_{\mathcal{U}_{\mathcal{D}}}(q_2)$ 
 and condition (ii) of Definition~\ref{Def:altsimu_up} is also satisfied.

 Now, for $(q_1,q_2) \in \mathcal{R}$ and $u_{q_2} \in U^a_{\mathcal{U}_\mathcal{D}}(q_2)$, we have that $\Delta_{\mathcal{U}_\mathcal{D}}(q_2,u_{q_2}) \subseteq X$.
For $u_{q_1} \in \{\downarrow u_{q_2}\} \cap U_{\mathcal{U}}$, due to the monotonicity of the system $\Sigma$, and since $x_2^{q_1} \leq x_2^{q_2} $ and $u_{q_1} \leq u_{q_2}$
we have for all $k \in \mathbb {K}^+(q_2,u_{q_2})$, $\tilde{x}_k' \geq f(x_2^{q_1},u_{q_1}, \overline{D})$.
Hence, $ f(x_2^{q_1},u_{q_1}, \overline{D}) \leq \max(\mathcal Q_u(q_2,u_{q_2})) \subseteq \{\downarrow X\} \subseteq X$,
 where the first inclusion comes from the definition of $\mathcal Q_u(q_1,u_{q_1})$ in (\ref{eqn:Q_u}) and the last inclusion comes from the lower closedness of the set $X$. Hence, there exists $q_{1}' \in X_{\mathcal{U}}$
such that $q_1' \in \Delta_{\mathcal{U}}(q_1,u_{q_1})$, which in turn implies that
$u_{q_1} \in U^a_{\mathcal{U}}(q_1)$. Finally, since $U_{\mathcal{U}_\mathcal{D}} = U_{\mathcal{U}}$ by construction, we have $U_{\mathcal{U}_\mathcal{D}} \subseteq \{\uparrow U_{\mathcal{U}}\}$, and condition (iii) of Definition~\ref{Def:altsimu_up} is satisfied.
 
For $(q_1,q_2) \in \mathcal{R}$ and $u_{q_2} \in U^a_{\mathcal{U}_\mathcal{D}}(q_2)$, consider $u_{q_1} \in U^a_{\mathcal{U}}(q_1)$
such that $u_{q_1} \leq u_{q_2}$, the existence of such $u_{q_1}$ is guaranteed in view of Remark~\ref{rk:existence}.
We have $x_2^{q_1'}$ is defined such that $f(x_2^{q_1},u_{q_1}, \overline{D}) \in [x_1^{q_1'},x_2^{q_1'}]$,
meaning that $x_2^{q_1'} \geq f(x_2^{q_1},u_{q_1}, \overline{D})$ and $\forall q_1^{\star} \in  X_{\mathcal{U}}$ such that
$x_2^{q_1^{\star}} \geq f(x_2^{q_1},u_{q_1}, \overline{D})$, then $x_2^{q_1^{\star}} \geq x_2^{q_1'}$.
For $q_2' = \Delta_{\mathcal{U}_\mathcal{D}}(q_2,u_{q_2})$, we have $ x_2^{q_2'} = \max \{x_2^{q^\star} \mid {q^\star} \cap \mathcal Q_u(q_2,u_{q_2}) \neq \emptyset \} $,
making $x_2^{q_2'} \geq \max(\mathcal Q_u(q_2,u_{q_2}))$. 
Due to the monotonicity of the system, and since $x_2^{q_1} \leq x_2^{q_2} $ and $u_{q_1} \leq u_{q_2}$
we have for all $k \in \mathbb {K}^+(q_2,u_{q_2})$, $\tilde{x}_k' \geq f(x_2^{q_1},u_{q_1}, \overline{D})$.
Hence, $\max(\mathcal Q_u(q_2,u_{q_2})) \geq f(x_2^{q_1},u_{q_1}, \overline{D})$.
Therefore, $x_2^{q_1'} \leq x_2^{q_2'}$, and  $(q_1',q_2') \in \mathcal{R}$ and condition (iv) of Definition~\ref{Def:altsimu_up} is satisfied.
\end{proof}

\begin{remark}
\label{rem:general_disturbances}
Although Theorem~\ref{thm:main_data} assumes that data points are collected under maximum disturbances, the result can be generalized to datasets collected under arbitrary disturbances. Specifically, if a Lipschitz constant $K_d > 0$ of the vector field $f$ with respect to the disturbance is known (i.e., $|f(x, u, d_2) - f(x, u, d_1)| \leq K_d |d_2 - d_1|$ for all $d_1, d_2 \in D$), one can define a bloating parameter $\tilde{\delta} := K_d |\overline{D} - \underline{D}|$ and modify the target sets $Q_u(q,u)$ and $Q_l(q,u)$ in \eqref{eqn:Q_u} and \eqref{eqn:Q_l} by adding and subtracting $\tilde{\delta} \cdot \mathbf{1}_n$, respectively. The same soundness relations in \eqref{eqn:soundness_model} then remain valid.
\end{remark}

\subsection{Controlling the Conservativeness of the Data-Driven Sparse Abstractions}
We now provide an approach to control the conservativeness between 
the model-based abstractions $\mathcal{L}_{\Sigma}, \mathcal{U}_{\Sigma}, $
and the data-driven abstractions $\mathcal{L}_{\Sigma_{\mathcal{D}}}, \mathcal{U}_{\Sigma_{\mathcal{D}}}$. 
Let us consider the perturbed version of the discrete-time control system $\Sigma$ defined in \eqref{dis_sys_pert}.
The objective of this section is to establish a relation between the data-driven abstraction of the system $\Sigma$ and the model-based abstraction of the perturbed system $\Sigma_{\varepsilon}$. To establish such a relation, we also assume prior knowledge of an upper bound on the derivative of the unknown function $f$ as stated in Assumption \ref{ass:bounds}.

Moreover, we assume the capability of sampling data points in small intervals in the neighborhood of grid points using all the possible inputs. First, we define those intervals on the upper side of the grid points $x_q^2$. Given $\xi \in \mathbb{R}^n_{\geq 0}$, for all $q \in \tilde X_{\mathcal{U}_{\mathcal{D}}}$, the sampling interval is defined as $\mathcal{I}_{q}^+(\xi) := [x_2^q, x_2^q + \xi] \subseteq \mathbb{R}^n$.

\begin{assumption}
\label{ass: sample_from_I}
    The data set contains data points from each interval $\mathcal{I}_{q}^+(\xi)$ using all the possible inputs (i.e., for all $q \in X_{\mathcal{U}_{\mathcal{D}}}$   
for all $u \in U_{\mathcal{U}_{\mathcal{D}}}$,
there exist $k \in K$ such that $\tilde{x}_k \in \mathcal{I}_{q}^+(\xi)$ and $\tilde{u}_k = u$)
\end{assumption}
For some $q \in \tilde X_{\mathcal{U}_{\mathcal{D}}}$ ,
the interval $\mathcal{I}_{q}^+$ is not included in the set $X$. Therefore, 
we will assume that we sample data points from the set $\tilde{X} $, such that
$X \subseteq \tilde{X}$ and 
$\mathcal{I}_{q}^+ \subseteq \tilde{X} $ for all $q \in \tilde X_{\mathcal{U}_{\mathcal{D}}}$. 
The next proposition shows that if we sample data points in the intervals $\mathcal{I}_{q}^+$, 
then we can control the conservativeness of the data-driven abstraction by 
establishing a relation between data-driven abstraction and the perturbed model-based abstraction.
\begin{theorem}
\label{the:precision}
Consider the discrete-time monotone control system \( S_\Sigma \) defined in \eqref{eqSsigma}, where the state space \( X \) is lower closed, the unknown map \( f \) satisfies Assumption~\ref{ass:bounds}, and Assumption~\ref{assum:contr_dist} holds. Let $\mathcal{U}_\Sigma$ in \eqref{eqn:up_abs}  be an upper-sparse 
abstraction of $S_\Sigma$ and $\mathcal{U}_{\Sigma_\mathcal{D}}$ in \eqref{eqn:up_abs_data} 
be a data-driven upper-sparse abstraction of $S_\Sigma$ such that Assumption \ref{ass33} holds. Furthermore, given $\varepsilon \in \mathbb{R}^n_+$, we let $\mathcal{U}_{\Sigma_{\varepsilon}}$ be upper-sparse abstraction of $S_{\Sigma_{\varepsilon}}$.
Then, for each $\varepsilon \in \mathbb{R}^n_+$ and for each $\xi=(\xi_1,\xi_2,\ldots,\xi_n)\in \mathbb{R}^n_{\geq 0}$ satisfying:
\begin{equation*}
    \label{eqn:eta_nu}
    \begin{aligned}
         \varepsilon_i  \geq \sum\limits_{j=1}^n\alpha_{ij} \xi_j\qquad \forall i \in \{1,2,...,n\}.
    \end{aligned}
\end{equation*}
If Assumption \ref{ass: sample_from_I} holds, then
\begin{equation}
    \label{eqn:soundness_precision}
     \mathcal{U}_{\Sigma_{\varepsilon}}  \preccurlyeq^0_u \mathcal{U}_{\Sigma_{\mathcal{D}}} \preccurlyeq^0_u \mathcal{U}_{\Sigma} \preccurlyeq^0_u S_{\Sigma}. 
    \end{equation}
\end{theorem}
\begin{proof}
The fact that $\mathcal{U}_{\Sigma_{\mathcal{D}}} \preccurlyeq^0_u \mathcal{U}_{\Sigma}$ was previously shown in Theorem~\ref{thm:main_data}.  

Lets us show that $\mathcal{U}_{\Sigma_{\varepsilon}}  \preccurlyeq^0_u \mathcal{U}_{\Sigma_{\mathcal{D}}}$. 
Consider the relation $\mathcal{R} \subseteq X_{\mathcal{U}_{\mathcal{D}}} \times X_{\mathcal{U}}$
defined by $(q_1,q_2) \in \mathcal{R}$ if and only if $x_2^{q_1} \leq x_2^{q_2}$. 
Let us show that $\mathcal{R}$ is an $0$-ASUAS relation from $\mathcal{U}_{\Sigma_{\varepsilon}}$ to $\mathcal{U}_{\Sigma_{\mathcal{D}}}$.
Since $X^o_{\mathcal{U}_{\mathcal{D}}}=X^o_\mathcal{U}$ 
condition (i) of Definition~\ref{Def:altsimu_up} is directly satisfied. 
Now, let $(q_1,q_2) \in \mathcal{R}$, we have $x_2^{q_1}\leq x_2^{q_2}$ 
and $H_{\mathcal{U}_{\mathcal{D}}}(q_1)=x_2^{q_1} \leq x_2^{q_2} = H_{\mathcal{U}}(q_2)$. 
Hence, condition (ii) in Definition~\ref{Def:altsimu_up} is satisfied. 

Consider $(q_1,q_2) \in \mathcal{R}$ and $u_{2} \in U^a_{\mathcal{U}}(q_2)$,
we have that $\Delta_{\mathcal{U}}(q_2,u_2) \subseteq X$.
Let $u_{1} \in \{\downarrow u_{2}\} \cap U_{\mathcal{U}_{\mathcal{D}}}$, 
we have $ x_2^{q_1'}= \max \{x_2^{q^\star} \mid {q^\star} \cap \mathcal Q_u(q_1,u_{1}) \neq \emptyset \}  \subseteq \{\downarrow X\} \subseteq X$,
 where the first inclusion comes from the fact that $\mathcal Q_u$ is always non empty due to the fact we can sample a data point $\tilde{x}_k \in \mathcal{I}_{q_1}^+$, and $\mathcal Q_u (q_1,u_{1}) \subseteq \{\downarrow X\} $ by construction.   Hence, there exists $q_{1}' \in X_{\mathcal{U}_{\mathcal{D}}}$
such that $x_2^{q_1'} \in q_{1}'$ and $q_1' = \Delta_{\mathcal{U}_{\mathcal{D}}}(q_1,u_1)$, which in turn implies that
$u_{1} \in U^a_{\mathcal{U}_{\mathcal{D}}}(q_1)$. Finally, since $U_{\mathcal{U}} = U_{\mathcal{U}_{\mathcal{D}}}$ by construction, we have $U_{\mathcal{U}} \subseteq \{\uparrow U_{\mathcal{U}_{\mathcal{D}}}\}$, and condition (iii) of Definition~\ref{Def:altsimu_up} is satisfied.

Consider $(q_1,q_2) \in \mathcal{R}$, $u_{2} \in U^a_{\mathcal{U}}(q_2)$ 
and any $u_{1} \in U^a_{\mathcal{U}_{\mathcal{D}}}$ satisfying $u_{1} \leq u_{2}$.
Let  $q_1' \in \Delta_{\mathcal{U}_{\mathcal{D}}}(q_1,u_1)$. We have the existence of $k \in \mathbb{K}$,
 $\tilde{x}_k \in \mathcal{I}_{q_1}^+$ and
\begin{equation}
    \label{eq:ineq1}
    \begin{aligned}
        f(\tilde{x}_k, {u}_1, \overline D) &\leq  f(x_2^{q_1} + \xi, u_1, \overline D),\\
        & \leq \max \{ f(x_2^{q_1}, u_1, \overline D) + \Omega_{\varepsilon}(0)\},\\
        & \leq \max \{ f(x_2^{q_2}, u_2, \overline D) + \Omega_{\varepsilon}(0)\}.
    \end{aligned}
\end{equation} 
The first and third inequalities come from the monotonicity of the function $f$, whereas the second inequality comes from equation (\ref{eqn:eta_nu}) and Lemma~\ref{lem:growth_bound}.
Since,  $\mathcal Q_u(q_1,u_1) \subseteq \downarrow f(\tilde{x}_k, {u}_1, \overline D) $ by construction, we have that $x_2^{q_1'} \leq x_2^{q_2'}$, which implies that 
$(q_1',q_2') \in \mathcal{R}$, and
condition (iv) of Definition~\ref{Def:altsimu_up} is satisfied.
\end{proof}

Theorem \ref{the:precision} shows how to reach a certain precision $\varepsilon$ when building the data-driven abstraction. This precision can be reached by choosing an adequate $\xi$ and sampling data points in each interval $\mathcal{I}_q^+, q \in \tilde X_{\mathcal{U}_{\mathcal{D}}}.$

To relax Assumption \ref{ass: sample_from_I}, we will study the probability that this will occur if we sample the data uniformly.

Assume that data points $(\tilde{x}_k, \tilde{u}_k)$ for $k \in \mathbb{K}$ are generated such that each $\tilde{x}_k$ is drawn independently and uniformly from a state set $\tilde{X}$, and each $\tilde{u}_k$ is drawn independently and uniformly from ${U}_{\mathcal{D}}$.
Under the uniform sampling assumption, the probability that 
a single sample $(\tilde{x}_k, \tilde{u}_k)$ satisfies that $\tilde{x}_k \in \mathcal{I}_q^+(\xi)$ and $\tilde{u}_k = u_p$ for a given pair $(q, u_p) \in X_{\mathcal{U}_D} \times U_{\mathcal{U}_D}$
is:
$p(q, u_p) := P(\tilde{x}_k \in \mathcal{I}_q^+(\xi) \text{ and } \tilde{u}_k = u_p) = \frac{\text{Vol}(\mathcal{I}_q^+(\xi))}{\text{Vol}(\tilde{X})} \cdot \frac{1}{|{U}_{\mathcal{D}}|}$.
Since this probability $p(q, u_p)$ is constant for all pairs $(q, u_p)$, let this constant probability be denoted by $p_s$.
\begin{proposition}
\label{pro:prop}
    Let $\beta \in (0,1)$ be a confidence parameter and $M := |X_{\mathcal{U}_D}| \cdot |U_{\mathcal{U}_D}|$.
If the number of samples $N = |\mathbb{K}|$ satisfies:
\begin{equation}
    \label{num_of_points}
    N \ge \frac{-\log(\beta) + \log(M)}{-\log(1-p_s)}
\end{equation}
then with probability at least $1-\beta$, the following relations hold:
$$\mathcal{U}_{\Sigma_{\varepsilon}} \preccurlyeq_{u}^{0} \mathcal{U}_{\Sigma_{D}} \preccurlyeq_{u}^{0} \mathcal{U}_{\Sigma} \preccurlyeq_{u}^{0} S_{\Sigma}$$
\end{proposition}
\begin{proof} Let us consider the event $\mathcal{E}$ that for every abstract state $q \in X_{\mathcal{U}_D}$ and every abstract input $u_p \in U_{\mathcal{U}_D}$, there exists at least one sampled data point $(\tilde{x}_k, \tilde{u}_k)$ from the $N$ samples such that $\tilde{x}_k \in \mathcal{I}_q^+(\xi)$ and $\tilde{u}_k = u_p$. This is the condition required for Theorem \ref{the:precision} to hold, which in turn leads to the claimed ASUAS relations:
\begin{multline}
        \mathcal{E} = \{ \forall q \in X_{\mathcal{U}_D}, \forall u_p \in U_{\mathcal{U}_D}, \exists k \in \mathbb{K} \text{ s.t.} \\
    (\tilde{x}_k \in \mathcal{I}_q^+(\xi) \text{ and } \tilde{u}_k = u_p) \}
\end{multline}
Let $\mathcal{E}_{q,u_p}$ be the event that for a specific pair $(q, u_p)$, there exists at least one sample satisfying the condition:
$$\mathcal{E}_{q,u_p} := \left\{ \exists k \in \mathbb{K} \text{ s.t. } (\tilde{x}_k \in \mathcal{I}_q^+(\xi) \text{ and } \tilde{u}_k = u_p) \right\}$$
The probability that a single sample $(\tilde{x}_k, \tilde{u}_k)$ satisfies the condition for a specific $(q, u_p)$ is $p_s$.
The probability that a single sample does \emph{not} satisfy the condition for $(q, u_p)$ is $1 - p_s$.
Since the $N$ samples are independent, the probability that \emph{none} of the $N$ samples satisfy the condition for $(q, u_p)$ (i.e., the event $\mathcal{E}_{q,u_p}$ does \emph{not} occur, denoted $\mathcal{E}_{q,u_p}^c$) is:
$$\mathbb{P}(\mathcal{E}_{q,u_p}^c) = (1 - p_s)^N$$

The event $\mathcal{E}$ occurs if $\mathcal{E}_{q,u_p}$ occurs for all $M = |X_{\mathcal{U}_D}| \cdot |U_{\mathcal{U}_D}|$ 
pairs of $(q, u_p)$.
We are interested in the probability $\mathbb{P}(\mathcal{E})$.
The complementary event $\mathcal{E}^c = \bigcup_{q,u_p} \mathcal{E}_{q,u_p}^c$ satisfies:
$$\mathbb{P}(\mathcal{E}^c) = \mathbb{P}\left(\bigcup_{q,u_p} \mathcal{E}_{q,u_p}^c\right) \le \sum_{q,u_p} \mathbb{P}(\mathcal{E}_{q,u_p}^c)$$
There are $M$ such terms in the sum. Hence,
$$\mathbb{P}(\mathcal{E}^c) \le M \cdot (1 - p_s)^N$$
The probability that event $\mathcal{E}$ occurs is $\mathbb{P}(\mathcal{E}) = 1 - \mathbb{P}(\mathcal{E}^c)$. Thus,
$$\mathbb{P}(\mathcal{E}) \ge 1 - M(1-p_s)^N$$
We want this probability to be at least $1-\beta$, hence one gets $1 - M(1-p_s)^N \ge 1-\beta$, which implies that $(1-p_s)^N \le \frac{\beta}{M}$. Taking logarithms on both sides (note that $\log(1-p_s)$ is negative since $0 < p_s \le 1$), we obtain
$$N \ge \frac{\log\left(\frac{\beta}{M}\right)}{\log(1-p_s)} = \frac{\log(\beta) - \log(M)}{-\left(-\log(1-p_s)\right)} = \frac{-\log(\beta) + \log(M)}{-\log(1-p_s)}$$
This provides the required condition for $N$. If this condition on $N$ is met, then with probability at least $1-\beta$, event $\mathcal{E}$ holds.
\end{proof}

For the data-driven abstraction to meet a given $\varepsilon$ precision at a  confidence level $\beta$, \eqref{num_of_points} in Proposition \ref{pro:prop} suggests valide number of data points.

\section{Numerical example} \label{sec:6}

\subsection{Model description and control objective}

We consider a vehicle moving along a straight road. The dynamical model  is adapted from~\cite{saoud2018contract} and given by 
\begin{equation}
\label{eqn:model}
m\dot{v}=\beta(u,v)=\left\{
\begin{array}{l c r}
u-f_{0}-f_{2}v^2  &\text{ if }& v>0\\
\max(u-f_0,0)  &\text{ if }& v=0,
\end{array}
\right.
\end{equation}
where $m>0$ is the vehicle's mass, $u$ is the net engine torque applied to the wheels, $v \geq 0$ represents the vehicle's velocity, and the term $f_0+f_2v^2$ includes the rolling resistance and aerodynamics. For this system, $u \in [\underline{U},\overline{U}]$ is the control input. Moreover, we include a leading vehicle whose velocity satisfies $d \in D=[d_1,d_2]$, which is considered as a disturbance. One can easily check that Assumption~\ref{assum:contr_dist} is satisfied. 

The complete dynamical representation of the system is given by:
\begin{equation}
\label{eqn4}
\left\{
\begin{array}{r c l}
\dot{h}&=&d - v\\
m\dot{v}&=&\beta(u,v),
\end{array}
\right.
\end{equation}
where $h$ is the relative distance between the leader and the follower. Based on this continuous-time system, we generate a discrete-time model using the sampling period $\tau=0.5 s$, while conserving the monotonicity property of the system with respect to the partial order $\leq_2$. That is,  $y:=(y_1,y_2) \leq_2  x:=(x_1,x_2)$ if and only if $y_1 \geq x_1$ and $y_2 \leq x_2$.

\begin{table}
	\caption{{Vehicle and safety parameters}} 
	\centering
	\begin{tabular}{|c|c|c|}
		\hline 
		Parameter & Value & Unit \\ 
		\hline 
		$M$ & $1370$ & $Kg$ \\ 
		
		$f_0$ & $51.0709$ & $N$ \\ 
		
		$f_2$ & $0.4161$ &  $Ns^2/{m^2}$\\ 
		
		$\underline{U}$&  $-4031.9$ &   $mKg/{s^2}$ \\
		
	$\overline{U}$&  $2687.9$&   $mKg/{s^2}$ \\
		
		$d_1$&  $10$&   $m$ \\
		
		$d'$&  $70$&   $m$  \\
		
		$v_{\max}$&  $15$ &   $m/s$  \\
        \hline
	\end{tabular}
	\label{table:parameters}
\end{table}

The objective is to compute a controller to ensure that the velocity $v$ remains between $0$ and $v_{\max}$, and the relative distance between the leader and the follower remains larger than $0$, while assuming that the velocity of the leader belongs to the set $D=[0,v_{\max}]$. This can be formlised as the safety specification $\phi:=((0,+\infty)\times [0,v_{\max}])^w$. Moreover, since the constraint $v \geq 0$ is directly satisfied from the model description in~(\ref{eqn:model}), the considered safety specification $\phi$is lower closed with respect to the partial order $\leq_2$. We recall that the maximal safety controller is the most permissive controller ensuring safety, in the sense that it admits every input admitted by any other safe controller~\cite{tabuada2009verification}. Furthermore, the domain of the maximal safety controller is the largest safe set that is invariant under the maximal safety controller.

\subsection{Numerical results}



\begin{figure}[!t]
	\begin{center}
		\includegraphics[width=0.9\columnwidth]{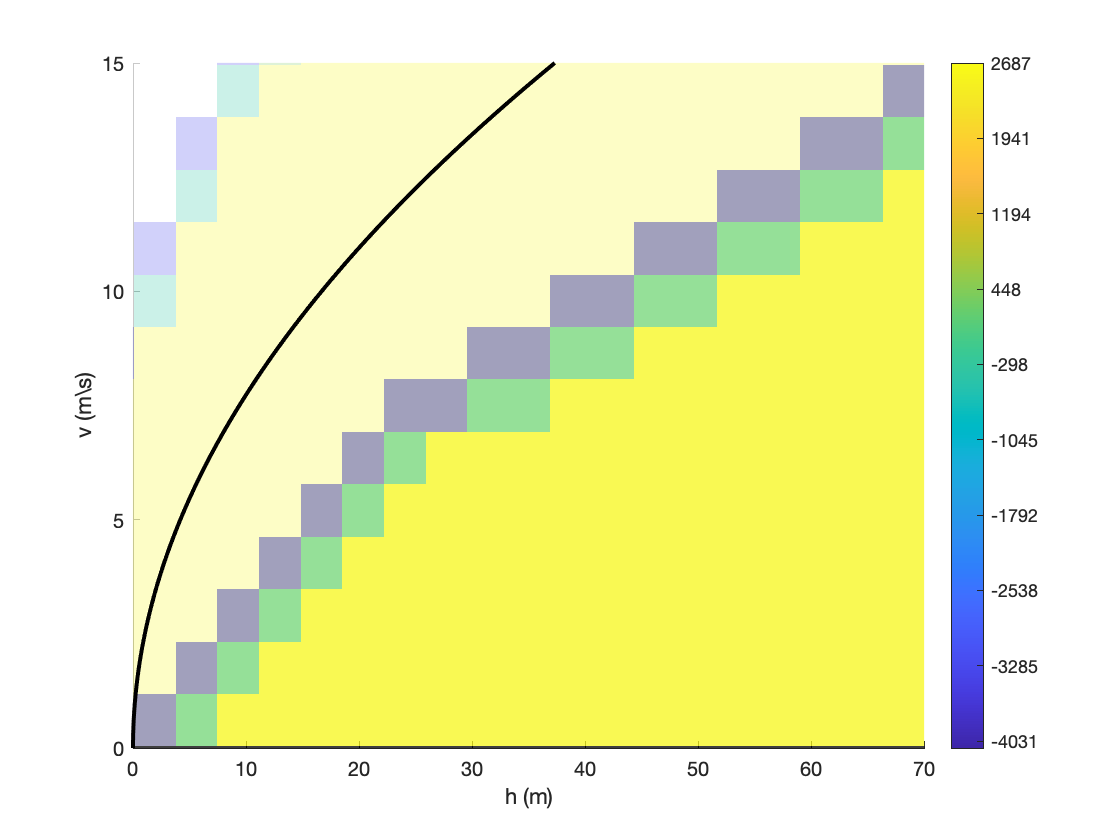}\\
	\end{center}
	\caption{Boundary of the domain of the maximal safety controller for the system $\Sigma$ together with the maximal safety controllers for the lower-sparse abstractions (larger safe region) and upper-sparse abstractions (smaller safe region), with $20\times 14$ discrete states and $10$ discrete inputs.}
	\label{fig:up_low_1}	
\end{figure}

In this section, we numerically illustrate the benefit of the proposed approach. The values shown in Table~\ref{table:parameters} are taken from~\cite{saoud2018contract}.

For the construction of the upper- and lower-sparse abstractions $\mathcal{L}_{\Sigma}$ and $\mathcal{U}_{\Sigma}$, we use the same partitioning technique presented in~\cite{saoud2018contract}, by generating a Cartesian partition in the state space $X$ and input space $U$.


Figure~\ref{fig:up_low_1} represents the symbolic controllers computed for the upper-sparse abstraction (smaller region) and the lower-sparse abstraction (larger region), with $20\times 14$ discrete states and $10$ discrete inputs. We used the lazy approach to compute controllers for monotone dynamical systems, as proposed recently in~\cite{ivanova2022lazy}. In the figure, the color bar represents the given input set $U=[\underline{U},\overline{U}]$, where the blue color corresponds to the minimal input $U_{\min}=-4031.9 \;mKg/{s^2}$, and the yellow color corresponds to the maximal input $U_{\max}=2687.9 \;mKg/{s^2}$. For a given state $(h,v)$ of the vehicle, Figure~\ref{fig:up_low_1} shows the maximal allowed control input. Let us mention that, following Theorem~\ref{thm:1}, and in view of the ASUAS relation proposed in Definition~\ref{Def:altsimu_up}, if an input $u \in U$ is allowed by the maximal safety controller, then all inputs satisfying $u' \leq u$ are allowed by that controller. Moreover, to show the precision provided by the proposed upper- and lower-sparse abstractions, we also present in black the boundary of the domain of the maximal safety controller for the original system $\Sigma$, which can be computed analytically for this problem, following the approach presented in~\cite{devonport2020data}.

In Figure~\ref{fig:up_low_1}, one can see that the domain of the maximal safety controller computed using the upper-sparse abstraction is included in the domain of the maximal safety controller for the original system $\Sigma$, which in turn is included in the domain of the maximal safety controller computed using the 
lower-sparse abstraction. This observation is consistent with the result of Theorem~\ref{thm:main}. Indeed, while the controller computed for the upper-sparse abstraction can be refined into a concrete controller for the original system $\Sigma$, the controller computed for the 
lower-sparse abstraction is used just to measure the conservativeness of our abstraction. Indeed, Figure~\ref{fig:up_low_1} shows that while the obtained controller for the upper-sparse abstraction is sound, in the sense that it can be refined into a controller for the original system $\Sigma$, the obtained controller is conservative. 

In order to measure the conservativeness of the obtained controller, and to improve its precision, we rely on the result of Theorem~\ref{thm:main2}. First, one can easily check that the values of the parameters $\alpha_{i,j}$, $i,j \in \{1,2\}$, given by $\alpha_{11}=\alpha_{12}=1$, $\alpha_{21}=0$, and $\alpha_{22}=1$ satisfy the inequalities in Assumption \ref{ass:bounds}. Hence, in view of condition (\ref{eqPrecision}), and since the discretization used corresponds to $20\times14$ discrete states, the precision $\varepsilon=(7,7)$ is achieved in that example, and one gets
$$\mathcal{L}_{\Sigma_{\varepsilon}}  \preccurlyeq^\eta_u \mathcal{U}_{\Sigma}  \preccurlyeq^0_u S_{\Sigma} \preccurlyeq^0_u \mathcal{L}_{\Sigma}.$$
To improve the result obtained in Figure~\ref{fig:up_low_1}, we now start by fixing a desired precision $\varepsilon=(0.5,0.5)$. In view of Theorem~\ref{thm:main2}, this precision can be achieved by choosing a discretization parameter $\eta=(0.25,0.25)$, for which, we used a partition corresponding to $(280 \times 60)$ discrete states. Figure~\ref{fig:up_low_2} represents the boundary of the domain of the maximal safety controller for $\Sigma$, together with the maximal safety controllers for the upper- and lower-sparse abstractions, with $280\times 60$ discrete states and $10$ discrete inputs. One can see in Figure~\ref{fig:up_low_2} that, for the chosen precision $\varepsilon=(0.5,0.5)$, the domain of the obtained controller using the 
upper-sparse abstraction is almost the same as the domain of the maximal safety controller, which is consistent with our theoretical results.

\begin{figure}[!t]
	\begin{center}
		\includegraphics[width=0.9\columnwidth]{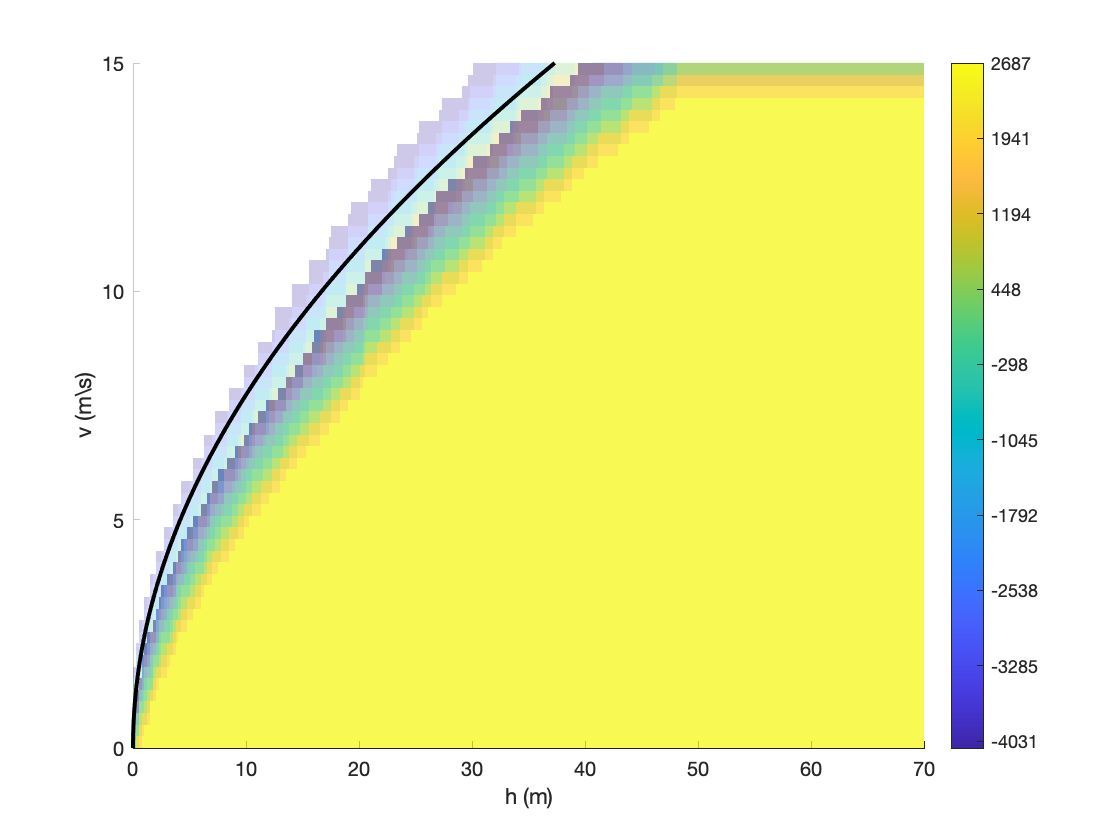}\\
	\end{center}
	\caption{Boundary of the domain of the maximal safety controller for the system $\Sigma$ together with the maximal safety controllers for the lower-sparse abstractions (larger safe region) and upper-sparse abstractions (smaller safe region), with $280\times 60$ discrete states and $10$ discrete inputs.}
	\label{fig:up_low_2}	
\end{figure}

\subsection{Solving the problem in 
3-dimensions}

In this section, we consider an adaptation of the model in (\ref{eqn4}). While in the latter case, the velocity of the leader was considered as a disturbance, in this case, we consider the velocity of the leader as a state variable, and we consider the acceleration of the leader 
to be the disturbance. 
The new dynamical model is given by
\begin{equation}
\left\{
\begin{array}{r c l}
\dot{h}&=&v_2 - v_1\\
m\dot{v}_1&=&\beta(u,v_2)\\
m\dot{v}_2&=&\gamma(d,v_1),
\end{array}
\right.
\end{equation}
where $h$ is the relative distance between the leader and the follower vehicles, $v_1$ is the velocity of the follower vehicle and $u \in [\underline{U},\overline{U}]$ is its control input. The map $\beta$ describing the dynamics of the follower is given in (\ref{eqn:model}). Similarly, $v_2$ is the velocity of the leader vehicle and $d \in D=[d_1,d_2]= [\underline{U},\overline{U}]$ is its acceleration, which is considered as a disturbance. The dynamics $\gamma$ of the leading vehicle is given by
\begin{equation*}
\label{eqn:model_leader}
m\dot{v_2}=\left\{
\begin{array}{l c r}
d-f_{0}-f_{2}v_2^2 &\text{if}& 0<v_2<v_{max}\\
\max(d-f_0,0) &\text{if}& v_2=0\\
\min(d-f_0-f_2 v_{max}^2,0)
&\text{if}& v_2=v_{max},
\end{array}
\right.
\end{equation*}
where $v_{\max}$ is given in Table \ref{table:parameters}. For the leader vehicle, we assume that the velocity $v_2$ remains in $[0,v_{\max}]$. Figure \ref{fig:upper_1_3d} shows the maximal allowed controller computed for the upper-sparse abstraction. Similarly, Figure \ref{fig:lower_1_3d} shows the maximal allowed controller computed for the lower-sparse abstraction. Both abstractions were computed with $20\times 14 \times 14$ discrete states and $10$ discrete inputs.

\begin{figure*}[!t]
    \begin{minipage}[t]{0.45\textwidth}
    \centering
    \includegraphics[width=\linewidth]{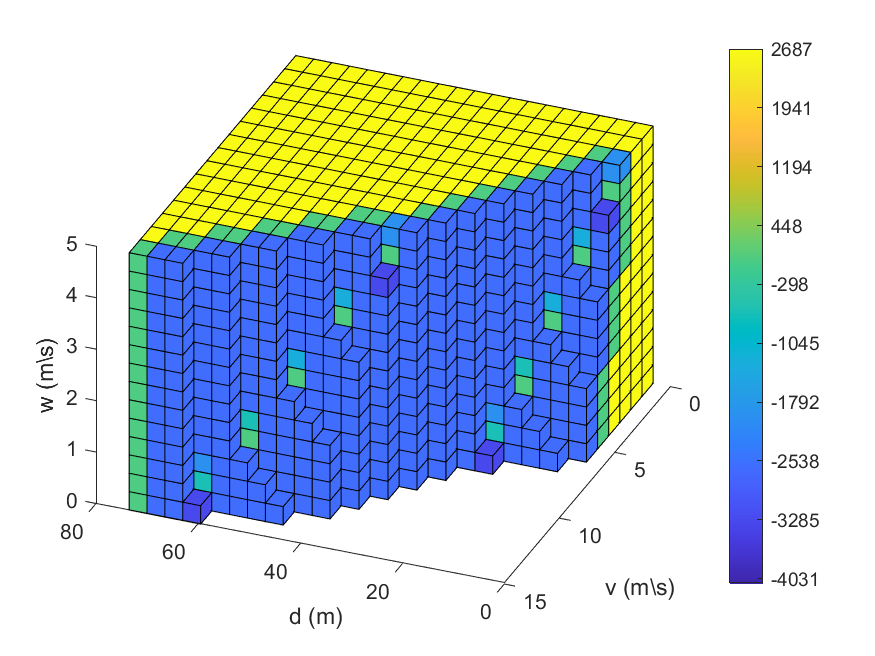}\
    \caption{Maximal safety controller computed for the upper-sparse abstraction $\mathcal{U}_{\Sigma}$, with $20\times 14 \times 14$ discrete states and $10$ discrete inputs.}
    \label{fig:upper_1_3d}
    \end{minipage}\hfill
    \begin{minipage}[t]{0.45\textwidth}
    \centering
    \includegraphics[width=\linewidth]{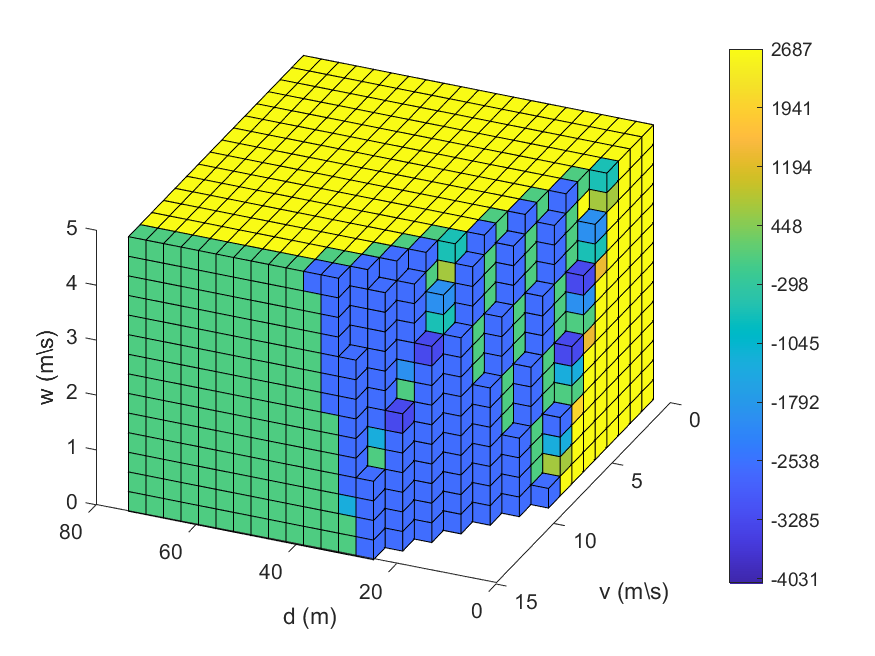}\
    \caption{Maximal safety controller computed for the lower-sparse abstraction $\mathcal{L}_{\Sigma}$, with $20\times 14 \times 14$ discrete states and $10$ discrete inputs.}
    \label{fig:lower_1_3d}
    \end{minipage}
\end{figure*}

\subsection{Numerical Results for Data-Driven Model}

This section demonstrates the data-driven approach by comparing the controller synthesized from an upper-sparse data-driven abstraction, $\mathcal{U}_{\Sigma_{\mathcal{D}}}$, with its model-based counterpart, $\mathcal{U}_{\Sigma}$, and the controller synthesized from a lower-sparse data-driven abstraction, $\mathcal{L}_{\Sigma_{\mathcal{D}}}$, 
with its model-based counterpart, $\mathcal{L}_{\Sigma}$. The data-driven abstractions are constructed from a dataset of $10^4$ points collected under the {maximum disturbance} scenario, as described in Section \ref{sec:dd_max}. Both the model-based and the data-driven abstractions are created using the same discretization of $70 \times 20$ grid for the states and $10$ inputs, to ensure a fair comparison.

The resulting maximal safety controllers 
for both abstractions are shown below. Figure~\ref{fig:notdata} displays the controller synthesized using the {model-based} and the {data-driven} upper-sparse abstractions, and we can see that the domain of the model-based controller contains the domain of the data-driven one, which is expected as according to {Theorem \ref{thm:main_data}}, the relationship between these two abstractions is $\mathcal{U}_{\Sigma_{\mathcal{D}}} \preccurlyeq^0_u \mathcal{U}_{\Sigma}$. This formal guarantee implies that the controller derived from data must be at least as conservative as the one derived from the model. This outcome successfully validates the data-driven approach, demonstrating its ability to produce a correct-by-construction controller that is provably sound, without any access to the underlying system equations.

Figure~\ref{fig:data} displays the controller synthesized using the model and {data-driven} lower-sparse abstractions. In this case, the domain of the controller derived from the data-driven abstraction contains the domain of the controller derived using the model of the system. This also validates Theorem \ref{thm:main_data}. 

\begin{figure*}[!t]
    \begin{minipage}[t]{0.48\textwidth}
    \centering
    \includegraphics[width=0.9\linewidth]{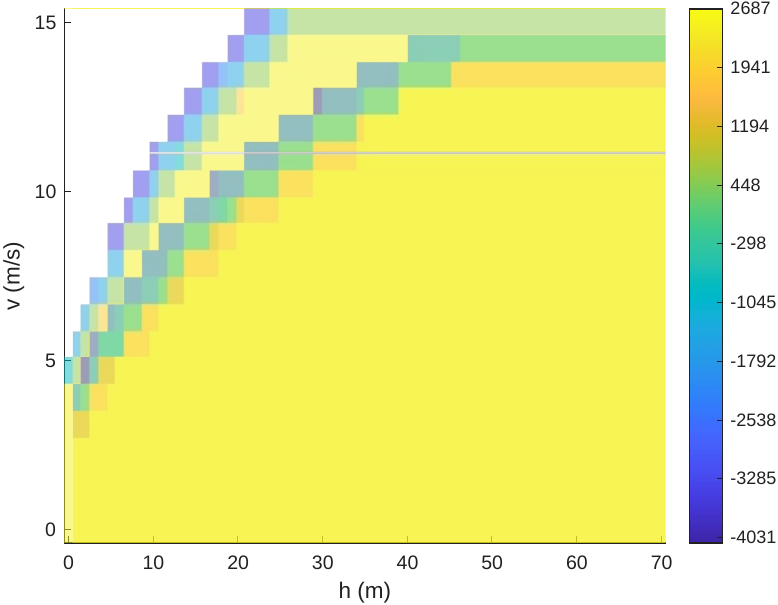}\
    \caption{Maximal safety controller for the {model-based} upper-sparse abstraction $\mathcal{U}_{\Sigma}$ (larger region) and the {data-driven} upper-sparse abstraction $\mathcal{U}_{\Sigma_{\mathcal{D}}}$ (smaller region). }
    \label{fig:notdata}
    \end{minipage}\hfill
    \begin{minipage}[t]{0.48\textwidth}
    \centering
    \includegraphics[width=0.9\linewidth]{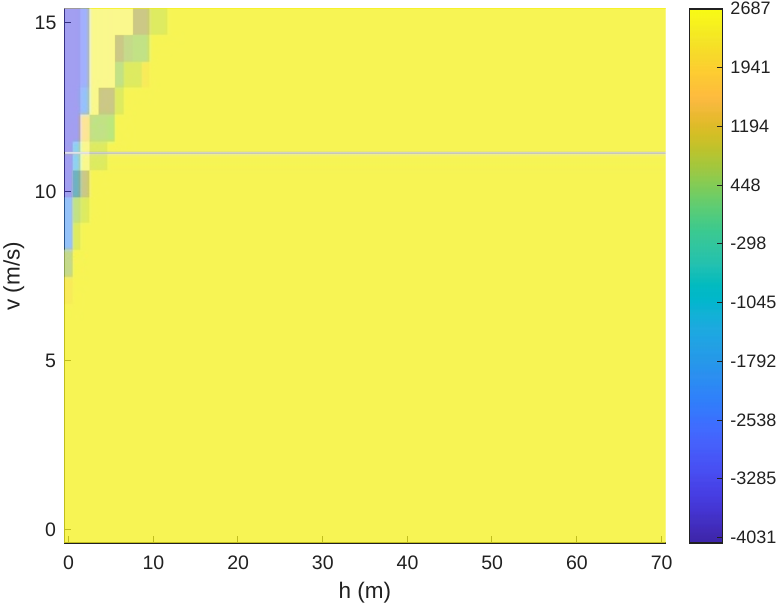}\
    \caption{Maximal safety controller for the {model-based} lower-sparse abstraction $\mathcal{L}_{\Sigma}$ (smaller region) and the {data-driven} lower-sparse abstraction $\mathcal{L}_{\Sigma_{\mathcal{D}}}$ (larger region).}
    \label{fig:data}
    \end{minipage}
\end{figure*}

\subsection{Validation of Conservativeness and Precision}
In this section, we experimentally validate the theoretical claims regarding the control of conservativeness (Theorem~\ref{the:precision} and Proposition~\ref{pro:prop}). All experiments were conducted using the same vehicle dynamics and the same $70 \times 20$ discretization grid.

\subsubsection{Validation of Theorem \ref{the:precision} (Deterministic Sampling)}
To validate Theorem~\ref{the:precision}, we construct a ``perfect'' dataset such that every abstract state-input pair $(q,u)$ has at least one sample in the critical region $\mathcal{I}_q^+(\xi)$. We use $\xi=[0.5, 0.5]$  
and computed the corresponding $\varepsilon=[1, 0.5]$. We then compute:
\begin{enumerate}
    \item A  model-based abstraction $\mathcal{U}_{\Sigma_{\varepsilon}}$, where the dynamics are expanded by $\varepsilon$.
    \item A data-driven abstraction $\mathcal{U}_{\Sigma_{\mathcal{D}}}$ using the perfect dataset.
\end{enumerate}
Since the dataset satisfies the premise of Theorem~\ref{the:precision}, the safe set of the model-based abstraction  $\mathcal{U}_{\Sigma_{\varepsilon}}$ is contained in that of the data-driven abstraction $\mathcal{U}_{\Sigma_{\mathcal{D}}}$, consistently with the relation $\mathcal{U}_{\Sigma_{\varepsilon}} \preccurlyeq^0_u \mathcal{U}_{\Sigma_{\mathcal{D}}}$. On the same discretization grid, the model-based abstraction yields $1172$ safe cells, while the data-driven abstraction yields $1200$. This is illustrated in Figure~\ref{fig:theorem5}.

\begin{figure}[!t]
	\begin{center}
		\includegraphics[width=0.9\columnwidth]{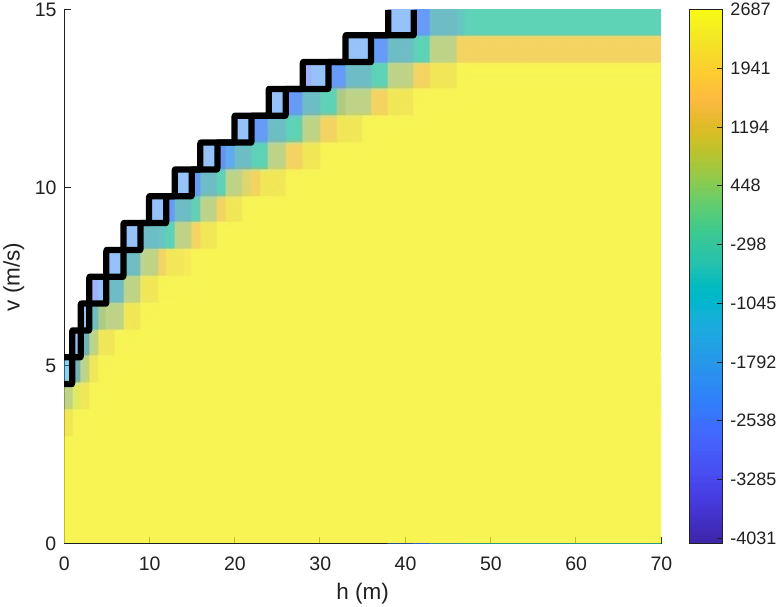}\\
	\end{center}
	\caption{Deterministic validation of Theorem 5. The robust model-based safe set is contained within the data-driven safe set, confirming $\mathcal{U}_{\Sigma_{\varepsilon}} \preccurlyeq \mathcal{U}_{\Sigma_{\mathcal{D}}}$.}
	\label{fig:theorem5}	
\end{figure}

\subsubsection{Validation of Proposition \ref{pro:prop} (Probabilistic Sampling)}
To validate the probabilistic bound in Proposition~\ref{pro:prop}, we evaluate the ability of random sampling to achieve the condition required for Theorem \ref{the:precision} on the same $70 \times 20$ grid with $10$ control inputs, resulting in $M = 14000$ total abstract state-input pairs. We used the same target region size $\xi=[0.5, 0.5]$ (with volume $\text{Vol}(\xi) = 0.25$) over the extended state space (with volume $\text{Vol}(\tilde{X}) = 1092.75$), which yields a single-sample hit probability $p_s = \frac{0.25}{1092.75 \times 10} \approx 2.2878 \times 10^{-5}$. Setting a confidence level $1-\beta = 0.90$, the required number of samples computed using equation \eqref{num_of_points} is $N \ge 517\,939$.
We performed a Monte-Carlo experiment with $1000$ independent trials. In each trial, we generate $N = 517\,939$ random samples and checked if every abstract state-input pair covers its corresponding $\xi$-region. The experiment achieved an empirical success rate of $918/1000$ ($91.8\%$), confirming the validity of the sampling formula.

\section{Conclusion}
This paper proposed a new approach to construct complete abstraction pairs for monotone control systems. A new behavioral relationship was introduced to relate the monotone system to its abstractions, and upper- and lower-sparse abstractions were designed in the context of the proposed relationships. Together, these two abstractions were shown to form a complete abstraction pair. Moreover, we presented an approach to control the conservativeness between the two abstractions. The results were further extended to data-driven systems, where the abstraction is constructed directly from a finite set of data points sampled from the system, rather than relying on an explicit model. Future directions include the generalization to  mixed-monotone systems, and improving scalability of the proposed approach using compositional techniques \cite{saoud2019compositional} and assume-guarantee contracts \cite{saoud2021assume}.

\appendix


\begin{lemma}
\label{lem:growth_bound}
If the map  $f:X \rightarrow X$ of the system $\Sigma$ in (\ref{dis_sys}) is continuously differentiable and satisfies
\begin{equation}
\label{eqn:bounds}
    0 \leq \frac{\partial f_i}{\partial x_j} \leq \alpha_{ij} ~~ \forall i,j \in \{1,2,\ldots,n\}
\end{equation}
then for any $x \in X$, for any $u \in U$, for any $d \in D$, and for any $\eta=(\eta_1,\eta_2,\ldots,\eta_n) \in \mathbb{R}^n_{+}$, we have:
$$\max\{\Omega_{\varepsilon}(f(x-\eta,u,d))\} \geq f(x,u,d) $$
for any $\varepsilon \in \mathbb{R}^n_{+}$ satisfying $\varepsilon_i \geq \sum\limits_{j=1}^n\alpha_{ij}\eta_j$, $i \in 
\{1,2,\ldots,n\}$.
\end{lemma}
\begin{proof}
From (\ref{eqn:bounds}), one has that for any $x \in X$, for any $u \in U$, for any $d \in D$ and for any $i \in \{1,2,\ldots,n\}$, $|f_i(x-\eta,u,d)-f_i(x,u,d)|  \leq \sum\limits_{j=1}^n\alpha_{ij}\eta_j \leq \varepsilon_i$. Hence, it follows that $f(x,u,d) \subseteq \Omega_{\varepsilon}(f(x-\eta,u,d))$
which in turn implies that 
$\max\{\Omega_{\varepsilon}(f(x-\eta,u,d))\} \geq f(x,u,d)$.
\end{proof}

\begin{lemma}
\label{lem:eps}
    Consider $A \subseteq \mathbb{R}^n$ and $\varepsilon \in \mathbb{R}^n_+$. Then,  $\downarrow \{\Omega_{\varepsilon}(A)\} \subseteq \Omega_{\varepsilon}(\downarrow A)$.
\end{lemma}
\begin{proof}
   Consider $x \in \downarrow\{\Omega_{\varepsilon}(A)\}$. Then there exist $z \in \Omega_{\varepsilon}(A)$ with $x \leq z$, and $a \in A$ such that $z \in \Omega_{\varepsilon}(a)$, i.e., $|z_i-a_i| \leq \varepsilon_i$ for all $i \in \{1,\ldots,n\}$. In particular, $x \leq z \leq a+\varepsilon$. Define $w \in \mathbb{R}^n$ component-wise by $w_i := \min\{a_i,\, x_i+\varepsilon_i\}$ for all $i \in \{1,\ldots,n\}$. Then $w \leq a$, hence $w \in \downarrow A$. Moreover, for each $i$ we have $w_i \leq x_i+\varepsilon_i$ by construction, and $w_i \geq x_i-\varepsilon_i$ since either $w_i = x_i+\varepsilon_i \geq x_i-\varepsilon_i$, or $w_i = a_i \geq x_i-\varepsilon_i$ (the latter because $x_i \leq a_i+\varepsilon_i$). Therefore $|x_i-w_i| \leq \varepsilon_i$ for all $i$, that is, $x \in \Omega_{\varepsilon}(w) \subseteq \Omega_{\varepsilon}(\downarrow A)$, which ends the proof.
\end{proof}





\begin{proposition}
\label{prop:4}
    Let $S_c:=(X_c,X^o_c,U_c,\Delta_c,Y_c,H_c)$ be a controller for the transition system $S_1=(X_1,X_1^o,U_1,\Delta_1,Y_1,H_1)$ such that $\mathcal{R}_{1,c}$ is an $\varepsilon$-approximate strong upper alternating simulation relation from $S_c$ to $S_1$, for some $\varepsilon \in \mathbb{R}^n_+$.
    Then there exists an $\varepsilon$-approximate upper simulation relation from
    $S_c \times_{\mathcal{R}_{1,c}} S_{1}$ to $S_c$. 
\end{proposition}
\begin{proof}
Consider the relation $\mathcal{R}$ defined as follows:
\begin{eqnarray*}
\mathcal{R} & =\{((x_1,x_c),\hat{x}_c)) \in X_1 \times X_c \times X_c \mid \\ & (x_1,x_c) \in \mathcal{R}\text{ and } x_c=\hat{x}_c\}.
\end{eqnarray*}
Let us show that $\mathcal{R}$ satisfies the requirement of Definition \ref{Def:simu_up}. First, since $ X_{1c}^o \subseteq X_1^o \times X_c^o$, condition (i) of Definition \ref{Def:simu_up} is directly satisfied. For $((x_1,x_c),\hat{x}_c)) \in \mathcal{R}$, since $(x_1,x_c) \in \mathcal{R}_{1,c}$ and $x_c=\hat{x}_c$, it follows that $H_{1c}(x_1,x_c)=H_1(x_1) \leq H_c(x_c)+\varepsilon = H_c(\hat{x}_c)+\varepsilon$, where the first inequality comes from the fact that $\mathcal{R}_{1,c}$ is an $\varepsilon$-ASUAS relation. Hence, condition (ii) of Definition \ref{Def:simu_up} is satisfied.

Consider $((x_1,x_c),\hat{x}_c)) \in \mathcal{R}$ and pick any $u \in U^a_{1c}(x_1,x_c)$ and $(x_1',x_c') \in \Delta_{1c}(x_1,x_c,u)$. It follows from the definition of the transition relation of the controlled system $S_c \times_{\mathcal{R}_{1,c}} S_{1}$ that $x_c' \in \Delta_c(x_c,u)$. Now we choose $\hat{u}=u$. Using the fact that $x_c=\hat{x}_c$, we have the existence of $\hat{x}_c' \in \Delta_c(\hat{x}_c,\hat{u})=\Delta_c(x_c,u)$ such that $\hat{x}_c'=x_c'$. Moreover, since $(x_1',x_c') \in \Delta_{1c}(x_1,x_c,u)$ it follows from definition of the transitions of the controlled system $S_c \times_{\mathcal{R}_{1,c}} S_{1}$ that $(x_1',x_c') \in \mathcal{R}_{1,c}$. Hence, we conclude that $((x_1',x_c'),\hat{x}_c')) \in \mathcal{R}$.
\end{proof}

\begin{proposition}
\label{prop:5}
    Consider the transition systems $S_i:=(X_i,X_i^o,U_i,\Delta_i,Y_i,H_i)$, $i=1,2$ such that $\mathcal{R}_{1,2}$ is an $\varepsilon$-approximate upper simulation relation from $S_1$ to $S_2$ for some $\varepsilon \in \mathbb{R}^n_{+}$. If $S_2$ satisfies a lower closed specification $\phi \subseteq Y_2^w$ ($\mathcal{B}(S_2) \subseteq \phi$),  then $S_1$ satisfies the specification $\Omega_{\varepsilon}(\phi)$. 
\end{proposition}
\begin{proof}
    Let $\sigma_{y_1} \in \mathcal{B}(S_1)$ be an output behaviour of the system $\Sigma_1$, and let $\sigma_1:=(x_{1,0},u_{1,0}),(x_{1,1},u_{1,1}),\ldots$ be the corresponding behavior. We now construct a behavior $\sigma_2:=(x_{2,0},u_{2,0}),(x_{2,1},u_{2,1}),\ldots$ of the system $S_2$ as follows:
    Since $x_{1,0} \in X^0_{1}$ and by definition of the approximate upper simulation relation $\mathcal{R}_{1,2}$, there exists $x_{2,0} \in X^0_{2}$ such that $(x_{1,0},x_{2,0}) \in \mathcal{R}_{1,2}$. It now follows again from the definition of approximate upper simulation relation, since $(x_{1,0},x_{2,0}) \in \mathcal{R}_{1,2}$ and $x_{1,1} \in \Delta_1(x_{1,0},u_{1,0})$, there must exist a transition $x_{2,1} \in \Delta_2(x_{2,0},u_{2,0})$ in $S_2$ with $(x_{1,1},x_{2,1}) \in \mathcal{R}_{1,2}$. We can repeat the argument again by using $(x_{1,1},x_{2,1}) \in \mathcal{R}_{1,2}$ and $x_{1,2} \in \Delta_1(x_{1,1},u_{1,1})$ to conclude the existence of $x_{2,2} \in \Delta_2(x_{2,1},u_{2,1})$ with $(x_{1,2},x_{2,2}) \in \mathcal{R}_{1,2}$. By repeating this argument inductively, one gets the existence of the behavior $\sigma_2$ satisfying $(x_{1,k},x_{2,k}) \in \mathcal{R}_{1,2}$ for all $k \in \mathbb{N}_0$. Finally, invoking condition (ii) of the definition of an approximate upper simulation relation, $(x_{1,k},x_{2,k}) \in \mathcal{R}_{1,2}$ implies $y_{1,k}=H_1(x_{1,k}) \leq y_{2,k}+\varepsilon=H_2(x_{2,k})+ \varepsilon$ for all $k \in \mathbb{N}$. Hence, one gets that $\mathcal{B}(S_1) \subseteq \downarrow\{\Omega_{\varepsilon}(\mathcal{B}(S_2))\}$. Moreover,  we also have that $\mathcal{B}(S_2) \subseteq \phi$. Combining the two last inclusions, one gets from Lemma \ref{lem:eps}, and using the fact that the set $\phi$ is lower closed that $\mathcal{B}(S_1) \subseteq \downarrow\{\Omega_{\varepsilon}(\mathcal{B}(S_2))\}
        \subseteq \downarrow\{\Omega_{\varepsilon}(\phi)\}  \subseteq \Omega_{\varepsilon}(\downarrow \phi) \subseteq \Omega_{\varepsilon}(\phi)$. Hence, one gets that $ \mathcal{B}(S_1) \subseteq \Omega_{\varepsilon}(\phi)$.
\end{proof}





\bibliographystyle{ieeetr}
\section*{References}
\bibliography{ref}

\end{document}